\documentclass[11pt, letterpaper]{article}
\pdfoutput=1
\usepackage{tikz}
\usepackage{bbm}
\usepackage{bm}
\usepackage{mathtools}
\usepackage{comment}
\usetikzlibrary{positioning,arrows,arrows.meta,automata,shapes.geometric,decorations,decorations.markings,backgrounds,calc}

\pgfdeclarelayer{background}
\pgfdeclarelayer{foreground}
\pgfsetlayers{background,main,foreground}

\usepackage{fullpage} % letter paper and margins
\usepackage{xspace,xcolor,graphicx,tabularx} %basic enhancements
\usepackage{mathtools,amsthm,amssymb} %basic maths packages
\usepackage{enumitem} %better lists
\setenumerate[0]{label={\normalfont (\roman*)}}

\usepackage[T1]{fontenc}
\usepackage[UKenglish]{babel}
\usepackage[scaled=0.86]{helvet}
\DeclareMathAlphabet{\mathcal}{OMS}{ntxm}{m}{n}
\usepackage{dsfont} %double stroke font
\let\mathbb\relax
\let\mathbb\mathds
\usepackage{stmaryrd} %extra symbols
\makeatletter
\renewcommand{\paragraph}{%
  \@startsection{paragraph}{4}%
  {\z@}{2.25ex \@plus 1ex \@minus .2ex}{-1em}%
  {\normalfont\normalsize\bfseries}%
}
\makeatother
\usepackage{authblk}

\definecolor{linkblue}{HTML}{001487}
\usepackage[colorlinks=true,allcolors=linkblue]{hyperref}
\usepackage{url}
\usepackage[nameinlink,capitalize,noabbrev]{cleveref}
\crefname{enumi}{Step}{Steps}

\newtheorem{theorem}{Theorem}[section]
\newtheorem*{theorem*}{Theorem}

\newtheorem{lemma}[theorem]{Lemma}

\newtheorem{fact}[theorem]{Fact}
\newtheorem{corollary}[theorem]{Corollary}
\theoremstyle{remark}

\theoremstyle{definition}
\newtheorem{definition}[theorem]{Definition}

\numberwithin{equation}{section}
\newtheorem{question}{Question}
\newtheorem*{question*}{Question}

\newcommand{\eps}{\epsilon}

\newcommand{\F}{\ensuremath{\mathds{F}}}
\newcommand{\N}{\ensuremath{\mathds{N}}}
\newcommand{\R}{\ensuremath{\mathds{R}}}

\DeclareMathOperator{\poly}{poly}
\DeclareMathOperator{\negl}{negl}

\newcommand{\ket}[1]{|#1\rangle}
\newcommand{\bra}[1]{\langle#1|}

\DeclarePairedDelimiterX\braket[2]{\langle}{\rangle}{#1 \delimsize\vert #2}

\newcommand{\cK}{\ensuremath{\mathcal{K}}}

\newcommand{\pt}{\mathsf{pt}}
\newcommand{\TN}{\mathsf{TN}}

\newcommand{\mincut}{\mathsf{mincut}}
\newcommand{\bF}{\bm{F}}
\newcommand{\bP}{\bm{P}}

\renewcommand{\H}{\mathbb{H}}
\newcommand{\TR}{\mathsf{tr}}

\newcommand{\Haar}{\mathsf{Haar}}
\newcommand{\State}{\mathsf{State}}
\newcommand{\caE}{\mathcal{E}}

\newcommand{\bulk}{\mathsf{bulk}}
\newcommand{\mc}{\mathsf{mc}}
\newcommand{\Weight}{\textsf{Weight}}

\newcommand{\nseg}{n_0}

\newcommand{\Ex}{\operatornamewithlimits{\mathbb{E}}}
\def\caD{\mathcal{D}}
\renewcommand{\N}{\mathbb{N}}
\newcommand{\posN}{\mathbb{N}_{\ge 1}}

\newcommand{\floor}[1]{\ensuremath{\left\lfloor#1\right\rfloor}}
\newcommand{\ceil}[1]{\ensuremath{\left\lceil#1\right\rceil}}
\def\caG{\mathcal{G}}

\newcommand{\spz}[1]{|#1\rangle}
\newcommand{\rpz}[1]{\langle #1 |}
\newcommand{\opt}[1]{\spz{#1}\rpz{#1}}
\newcommand{\optk}[2]{\left(\opt{#1}\right)^{\otimes #2}}
\newcommand{\WT}{\widetilde}

\newcommand{\mathdot}{\ensuremath{\;\text{.}}}
\newcommand{\mathcomma}{\ensuremath{\;\text{,}}}

\def\caH{\mathcal{H}}
\def\caT{\mathcal{T}}

\newcommand{\dmax}{d_{\sf max}}

\newenvironment{reminder}[1]{\bigskip
	\noindent {\bf Reminder of #1.  }\em}{\smallskip}

\newenvironment{proofof}[1]{\begin{proof}[{Proof of #1}]}{\end{proof}}

\begin{document}

\title{

Holographic pseudoentanglement and the \\ complexity of the AdS/CFT dictionary \vspace{5pt}}

\author[1]{Chris Akers}
\author[2]{Adam Bouland}
\author[3]{Lijie Chen}
\author[2]{Tamara Kohler}
\author[4]{\\Tony Metger}
\author[3]{Umesh Vazirani}

\affil[1]{University of Colorado Boulder}
\affil[2]{Stanford University}
\affil[3]{UC Berkeley}
\affil[4]{ETH Zurich}

\date{}

\maketitle

\begin{abstract}
The ``quantum gravity in the lab'' paradigm suggests that quantum computers might shed light on quantum gravity by simulating the CFT side of the AdS/CFT correspondence and mapping the results to the AdS side.
This relies on the assumption that the duality map (the ``dictionary'') is efficient to compute.
In this work, we show that the complexity of the AdS/CFT dictionary is surprisingly subtle: there might be cases in which one can efficiently apply operators to the CFT state (a task we call ``operator reconstruction'') without being able to extract basic properties of the dual bulk state such as its geometry (which we call ``geometry reconstruction'').
Geometry reconstruction corresponds to the setting where we want to extract properties of a completely unknown bulk dual from a simulated CFT boundary state.

We demonstrate that geometry reconstruction may be generically hard due to the connection between geometry and entanglement in holography. 
In particular we construct ensembles of states whose entanglement approximately obey the Ryu-Takayanagi formula for arbitrary geometries, but which are nevertheless computationally indistinguishable.
This suggests that even for states with the special entanglement structure of holographic CFT states, geometry reconstruction might be hard.
This result should be compared with existing evidence that operator reconstruction is generically easy in AdS/CFT.
A useful analogy for the difference between these two tasks is quantum fully homomorphic encryption (FHE): this encrypts quantum states in such a way that no efficient adversary can learn properties of the state, but operators can be applied efficiently to the encrypted state. We show that quantum FHE can separate the complexity of geometry reconstruction vs operator reconstruction, which raises the question whether FHE could be a useful lens through which to view AdS/CFT.

% A potentially limiting aspect of this program is that in certain cases the AdS/CFT dictionary connecting these theories can be exponentially hard to compute.
% In this work we examine the complexity of the dictionary through the lens of geometry reconstruction, where the goal is to learn the AdS geometry from (many copies) of an unknown CFT state - modeling what might occur if one simulates a CFT on a quantum computer.
% We argue this task might be generically hard due to the connection between geometry and entanglement in holography.
% In particular we construct ensembles of states whose entanglement approximately obey the Ryu-Takayanagi formula for arbitrary geometries, but which are nevertheless computationally indistinguishable.
% This opens the question whether the AdS/CFT dictionary could be exponentially hard to compute outside of black holes or horizons, a task which had previously been thought to be easy in a more limited setting which we call operator reconstruction.
% We show that quantum fully homomorphic encryption schemes (FHE) can separate the complexity of geometry reconstruction vs operator reconstruction, which raises the question whether homomorphic encryption could be a useful lens through which to view AdS/CFT.

\end{abstract}

\newpage
\section{Introduction}

A central challenge of theoretical physics is to develop a unified theory of quantum gravity.
A major source of progress in this area has been the AdS/CFT correspondence~\cite{Maldacena_1999} -- a conjectured duality between a theory of quantum gravity in Anti de Sitter (AdS) space, and a conformal field theory (CFT) living on its boundary.  
The two theories are connected by a ``dictionary'' which maps states and operators in one theory to the other, allowing us to learn about quantum gravity by studying the dual system using tools from quantum information.
This has led to number of key insights into quantum gravity, such as a sharper understanding of spacetime as an emergent phenomenon~\cite{Ryu_2006, VanRaamsdonk:2010pw, Almheiri:2014lwa} and progress towards solving the black hole information paradox~\cite{penington2020entanglementwedgereconstructioninformation,Almheiri_2019, akers2022blackholeinteriornonisometric}.
It also raises the exciting possibility that quantum computers might one day shed light on quantum gravity, by simulating the dual quantum mechanical system.
This ``quantum gravity in the lab'' paradigm~\cite{Brown_2023,Nezami_2023} has already seen its first toy experimental implementations~\cite{Jafferis_2022, Shapoval:2022xeo}.

However, realizing this potential requires not only that the dictionary exists, but that we have an explicit and \emph{efficiently computable} description of it -- otherwise we can only learn from the dual system in principle, but not in practice.
In this paper, we argue that the complexity of the AdS/CFT dictionary is surprisingly subtle and depends crucially what we mean by ``implement the dictionary''.
We distinguish between two different versions of this question:
\begin{itemize}
    \item \emph{Operator reconstruction:} 
    Given an AdS operator, how complicated is the CFT operator that ``reconstructs'' it? 
    More precisely, let $\mathcal{H}_\mathrm{code}$ be a subspace of AdS states, $V: \mathcal{H}_\mathrm{code} \to \mathcal{H}_\mathrm{bdry}$ the linear (``holographic'') map to the Hilbert space of the CFT dual, and $U$ an operator on $\mathcal{H}_\mathrm{code}$. 
    The goal of operator reconstruction is to implement some $U_\mathrm{bdry}$ on $\mathcal{H}_\mathrm{bdry}$ such that
    \begin{equation}
        U_\mathrm{bdry} V \ket{\psi} \approx V U \ket{\psi}~,
    \end{equation}
    for any $\ket{\psi} \in \mathcal{H}_\mathrm{code} \otimes \mathcal{H}_R$ where $R$ is an arbitrary reference system.
    Note that $U_\mathrm{bdry}$ is specific to the code subspace, which means that it can depend on properties of the bulk such as its geometry.

    \item \emph{Geometry reconstruction:} given (possibly many copies of) a CFT boundary state, how hard is it to estimate properties of the geometry of the dual gravitational system, promised a simple semiclassical geometrical dual exists? This is a simplified version of \emph{state reconstruction} -- where the goal is to produce the AdS state from the boundary state -- as we will focus on properties easy to compute for an AdS observer, such as the spacetime curvature in some subregion.
\end{itemize}

The goal of this paper is to study whether operator and geometry reconstruction may have different complexity in AdS/CFT, and to provide examples where geometry reconstruction is hard in toy models of AdS/CFT even outside of the analog of horizons.
We summarize this by the following question:

\begin{question*}
Are there scenarios in which operator reconstruction is easy but geometry reconstruction is hard? Are these scenarios relevant in real AdS/CFT, and what limits do they place on the ``quantum gravity in the lab'' paradigm? 
\end{question*}

At first glance, reconstructing operators vs states looks like two sides of the same coin, related by switching between the Heisenberg and Schr\"odinger picture.
However, using ideas from quantum cryptography, in particular pseudoentanglement and quantum fully homomorphic encryption, we will argue that the complexities of operator and geometry reconstruction can differ, even in scenarios that can reproduce some aspects of real AdS/CFT. 

\subsection{Holographic pseudoentanglement}
A core tenet of AdS/CFT is that the entanglement of CFT boundary states is dual to the geometry of their AdS bulk duals.
This is exhibited by the Ryu-Takayanagi (RT) formula which states that the entanglement entropy of the CFT state is related to the length of minimal geodesics in the AdS space~\cite{Ryu_2006}.\footnote{{Of course, it is well understood that knowing the geometry is not sufficient to compute the boundary entanglement in general, because in many cases one must use the \emph{quantum extremal surface} formula \cite{Engelhardt:2014gca}, which requires you know the bulk entropy. However, in this paper we will restrict our attention to setups in which the simpler Ryu-Takayanagi formula holds, i.e.~all boundary entropies are computed by the minimal \emph{area} surface in the bulk. We will argue that even when restricting to these situations, distinguishing the bulk geometry can be hard in toy models.}}
This connection seems odd from a computer science perspective, because entanglement is in general not an efficiently measurable property of quantum states.
In fact, it has recently been shown that even exponentially large gaps in entanglement can be cryptographically hard to detect, a phenomenon known as pseudoentanglement~\cite{aaronson2023quantumpseudoentanglement,gheorghiu2024estimatingentropyshallowcircuit}.
This suggests that geometry reconstruction might be difficult for generic states in AdS/CFT, as an efficient algorithm for reconstructing the geometry of a holographic state would provide an efficient method of  calculating the entanglement, which the existence of pseudoentangled states demonstrates is not possible in general.
However, this line of argument  falls short of showing hardness of geometry reconstruction.
This is because holographic states arising in the CFT duals of smooth AdS geometries have very particular entanglement structures.
In particular, they live in the ``holographic entropy cone'', which is a set of entropy inequalities imposed by the RT formula~\cite{Bao:2015bfa}.
Prior constructions of pseudoentanglement did not obey the RT formula for any geometry, and therefore this argument did not connect to AdS/CFT.

In this work we close this gap by constructing pseudoentanglement approximately within the holographic entropy cone. In particular we show that it can be hard to distinguish states which approximately obey the RT formula for arbitrarily different geometries:

\begin{theorem}[Informal] For any two bulk geometries $g_1,g_2$, there exist two ensembles of quantum states $\{\ket{\Psi}_k$,$\ket{\Phi}_k\}_{k \in \cK}$ such that
\begin{enumerate}
    \item The states $\ket{\Psi}_k$,$\ket{\Phi}_k$ are efficiently constructable by poly-size quantum circuits given the key $k$.
    \item The states $\ket{\Psi}_k$,$\ket{\Phi}_k$ approximately obey the RT formula for geometries $g_1,g_2$ respectively (for any choice of key $k$).
    \item No poly-time quantum algorithm can distinguish a random $\ket{\Psi}_k$ from a random $\ket{\Phi}_k$ given polynomially many copies of the state.
\end{enumerate}
\end{theorem}

In other  words, we can create holographic pseudoentangled states with any two geometries we desire.
For example, one ensemble could correspond to a bulk with a black hole, and the other without, and given copies of the CFT state, no algorithm can efficiently distinguish the two.
We also give a second construction of holographic pseudoentanglement, where the set of geometries spoofed is more limited, but the states exactly obey the RT formula. 
Our construction uses a discrete toy model of gravity based on tensor networks \cite{Pastawski_2015,Hayden_2016}.
We note that the concrete constructions we provide here suffer the same limitations of other tensor network toy models of AdS/CFT, and while our boundary states obey the Ryu-Takayanagi formula they do not necessarily have all other properties CFT states.
However, one of our constructions could be applied more generally, and could be applied to CFT states -- we use tensor network toy models here in order to present a concrete construction. 
Our constructions only require the existence of a post-quantum secure one-way function, a standard cryptographic assumption.

Our result shows that geometry reconstruction is computationally hard even for states which obey the RT formula. 
This result applies to any pair of generic bulk geometries; we do not need any particular geometrical structure such as black holes or wormholes inside the bulk to argue that geometry reconstruction is hard. 
This might, at first sight, appear to be in contradiction to results in AdS/CFT showing that bulk reconstruction should be easy in the absence of horizons or similar geometric obstructions~\cite{brown2019pythonslunchgeometricobstructions,Engelhardt_2022}. 
We argue that this apparent contradiction can be due to the distinction between operator reconstruction and geometry reconstruction, and more subtly due to a difference in the input-output model considered in the two questions.

In the existing literature, both types of reconstruction have been considered, but their difference has not been made explicit.
In~\cite{bouland2019computationalpseudorandomnesswormholegrowth} the authors consider the task of determining the volume of a wormhole in AdS given access to the boundary CFT.
This is an example of geometry reconstruction in our language.
It was shown that this task is intractable on a quantum computer under plausible cryptographic assumptions by relating it to breaking quantum pseudorandomness constructions.
This led to a number of works studying when implementing the dictionary \emph{is} efficient~\cite{susskind2020blackholesexptime,susskind2020horizonsprotectchurchturing,engelhardt2024cryptographic,Engelhardt_2022}.
However, these works generally considered operator reconstruction.
Indeed, there came to exist a powerful, plausible conjecture about operator reconstruction, called the \emph{Strong Python's Lunch Conjecture}, which delineates exactly which AdS operators have high complexity reconstructions~\cite{brown2019pythonslunchgeometricobstructions,Engelhardt_2021}.
See \cref{sec:pythons lunch} for an extended discussion of this conjecture.
It turns out~\cite{Engelhardt_2022} that the Python's Lunch conjecture agrees with~\cite{bouland2019computationalpseudorandomnesswormholegrowth} -- apparently both operator and state reconstruction are hard inside the wormholes of~\cite{bouland2019computationalpseudorandomnesswormholegrowth}.
This sharpens the question: do operator and state reconstruction \emph{always} have the same complexity in AdS/CFT?

Our pseudoentanglement result suggests that the answer might be no: geometry reconstruction might be hard in cases even where operator reconstruction is easy.
Another way of looking at this distinction is that geometry and operator reconstruction consider different input-output models.

In operator construction, the goal is to \emph{implement} the boundary dual of a bulk operator, but the algorithms typically assume one is given as input substantial information about $\mathcal{H}_\mathrm{code}$, often even restricting to one particular fixed geometry.
 For example, the HKLL bulk reconstruction \cite{Hamilton_2006} assumes the geometry is known and fixed. 
 Similarly, in Python's lunch examples in tensor networks, one assumes the tensor network (and hence geometry, and even the values of the tensors) is given as input, and Python's lunch is formulated in studying the complexity of ``pushing'' operators on the bulk to the boundary through the tensor network.

In contrast, an algorithm for geometry reconstruction takes as input (polynomially many copies of) an $n$-qubit CFT state with an unknown geometry, and the entire goal is to learn the geometry of the corresponding state.  
This viewpoint is trying to generalize the tensor network toy model of gravity \cite{Pastawski_2015} to be one step closer to real AdS/CFT, where the dictionary should hold not just for one geometry, but for many.
In the tensor network toy model, this is akin to not knowing the geometry of the tensors in advance, and perhaps even their values. 
We argue that geometry reconstruction is a relevant question for quantum gravity in the lab, where the goal should be to simulate CFT states with a quantum computer for which \emph{we have little information about the dual bulk state ahead of time} -- if geometry reconstruction is hard, the quantum computer will be greatly limited in what \emph{new information} it can tell us.

\subsection{Intuition: separating state vs. operator complexities via FHE}

Before describing our construction of holographic pseudoentanglement, it will be helpful to first build some intuition for how  operator questions can be easy, but state questions can be hard. 
This is counterintuitive as one usually thinks of state and operators as being on the same footing.
An illustrative example comes from \emph{homomorphic encryption}.
The goal of fully homomorphic encryption (FHE) is to encrypt data in such a way that an adversary cannot read the data, but nonetheless can perform computations on it, i.e.~for any function $f$ transform an encryption of a string $x$ into an encryption of the string $f(x)$, without ever knowing $x$.
Classical FHE was famously constructed by Gentry~\cite{gentry2009fully}, and Mahadev has constructed a quantum version of FHE~\cite{Mahadev:2017bwz}.
More precisely, in a \emph{quantum} fully homomorphic encryption (QFHE) scheme, an encoding circuit $V_k$ (indexed by some key $k$) is applied to a quantum state $\ket{\psi}$.
Then, anyone who has access to the state $V_k \ket{\psi}$, but does not know the key $k$, cannot efficiently compute information about the original state $\ket{\psi}$; the state $\ket{\psi}$ has been encrypted.
The special property of a \emph{homomorphic} scheme is that nevertheless, someone with access to $V_k\ket{\psi}$ (and not $k$ itself) can still efficiently apply operations to the encoded state.
That is, for any $U$ we wish to apply to $\ket{\psi}$, it is easy to apply a $\widetilde{U}$ such that
$$
    \widetilde{U} V_k \ket{\psi} = V_k U \ket{\psi}. 
$$
The interesting property of this scheme is that we have enough knowledge about $V_k$ to efficiently apply operators, but not enough knowledge to efficiently learn about the underlying state.

The key point is that the problem of applying a unitary $U$ homomorphically on the encrypted state looks \emph{precisely} like the problem of operator reconstruction, where the AdS/CFT dictionary $V$ is playing the role of the encrypting map. This is, of course, not a perfect equivalence. Traditionally, it is assumed that the holographic map $V$ is some linear map that is fully known, i.e.~there is no secret key $k$. However, we cannot simply use a known $V_k$ from a QFHE scheme while keeping state reconstruction hard. We fix this issue in section \ref{sec:FHE_holo_maps}, where we use QFHE to construct a (fully known) linear map $V$ with the properties that state reconstruction is hard while operator reconstruction is easy for many operators.

This FHE construction provides a clean separation between the complexity of operator versus state reconstruction in the general case. However, it only does so in the \emph{single copy} setting.
Known examples of QFHE schemes are not secure against having multiple copies -- access to multiple copies of the state gives the power to make state reconstruction \emph{easy}.
The multi-copy setting is more relevant to the quantum gravity in the lab paradigm, as one could prepare multiple copies of the CFT state on a quantum computer. 
For this reason we have also included holographic pseudoentanglement constructions, which \emph{are} secure against multiple copies. However, the gap in complexity between state and operator reconstruction does not manifest as clearly in the holographic pseudoenanglement constructions. 
In \cref{sec:pythons lunch} we discuss in detail the gap in complexity between state and operator reconstruction in these constructions.

\subsection{Proof sketch for holographic pseudoentanglement}

We give two constructions of holographic pseudoentanglement.
The first, \emph{holographic pseudoentanglement from low-entangling pseudorandom unitaries (PRUs)} (\cref{sec:pe_pru}), is more flexible in the sense that it works for arbitrary geometries, i.e., we can use it to show that the geometry reconstruction problem is hard for any pair of geometries $g_1, g_2$.
The second, \emph{holographic pseudoentanglement from pseudoentangled link states} (\cref{sec:pe_link}) only hides more minor differences in geometry, but has the advantage that results in states that exactly obey the Ryu-Takayanagi formula, while the PRU example results in an approximate version of the formula.
The latter construction can also be made public-key, i.e., it remains secure if the state preparation circuits are known.
Both constructions are very simple.
Here, we give a brief overview and refer to \cref{sec:pe_pru} and \cref{sec:pe_link} for details.

\paragraph{Holographic pseudoentanglement from low-entangling PRUs.}
There exist many tensor network constructions of quantum states with (approximate) RT entanglement scaling, for example from perfect tensors, random tensor networks, or Clifford tensor networks~\cite{Pastawski_2015,Hayden_2016,Apel_2022}.
We will start from a pair of states arising from such constructions for two different geometries $g_1, g_2$.
Indeed, for us it is not even necessary that these ``starting states'' be constructed using a tensor network -- this merely serves to make the construction concrete, but we can use any pair of CFT states (modelled as a quantum state of $N$ systems with local dimension $d$) as our starting states.

We now want to hide the difference between these two states without altering their geometry too much, i.e.,~we want to apply some operation to these states that makes them indistinguishable, but preserves their RT entanglement structure.
For this, we rely on a recent result by Schuster, Huang, and Haferkamp~\cite{brickwork_t_design}, who proved that a two-layer brickwork arrangement of ``small'' Haar random unitaries is a good approximation to a ``big'' Haar random unitary.
They then observed that if one replaces the Haar random unitaries with pseudorandom unitaries (PRUs) on polylogarithmically many qubits (which have been recently constructed~\cite{metger2024simpleconstructionslineardepthtdesigns,chen2024efficientunitarydesignspseudorandom}), then one obtains PRUs with polylogarithmic circuit depth.\footnote{A PRU is an ensemble of unitaries that is efficiently implementable but computationally indistinguishable from a Haar random unitary -- see \cref{sec:prus} for a rigorous definition.}
As a matter of fact, these PRUs are not merely low-depth, but also low-entangling: due to their brickwork structure, the lightcone of any qubit is only polylogarithmically large.
This allows us to bound the change in entanglement structure produced by these PRUs in a straightforward manner and show that the states after applying the brickwork PRU still approximately satisfy the RT formula.

The final construction is shown in \cref{fig:hqecc pru}: we start from any two geometries, use tensor networks to obtain states whose entanglement structure obeys the RT formula for the chosen geometry, and then hide the difference between the two states by applying a 2-layer brickwork arrangement of polylogarithmically sized PRUs. 

\paragraph{Holographic pseudoentanglement from pseudoentangled link states.} For our second construction, we start with a \emph{tree} tensor network consisting of perfect tensors. Here, a perfect tensor (see~\cref{sec:prelim:perfect-tensors} for a formal definition) is a tensor with an even number of legs that acts as a unitary from any set of half its legs to the complement set. 
Let the tree $T$ have $n$ leaves (i.e.~boundary nodes) and local dimension $d = n^{\omega(1)}$.

In order to construct two different (computationally indistinguishable) geometries from this tree tensor network we will replace one of the links in the tensor network with a state from a \emph{pseudoentangled state ensemble}.
A pseudoentangled state ensemble is a pair of ensembles of quantum states, $\caD_{\sf low}$ and $\caD_{\sf high}$, such that all states from $\caD_{\sf low}$ (resp.~$\caD_{\sf high}$) have low (resp.~high) von Neumann entropy across a cut, and no poly-time quantum algorithm given polynomially many copies of a state can determine which distribution it was drawn from (see \cref{sec:pseudoentanglement} for a rigorous definition).
The pseudoentangled state ensemble we use is a bipartite quantum system of two $d$-dimensional qudits from \cite{aaronson2023quantumpseudoentanglement}, where $\caD_{\sf low}$ (resp.~$\caD_{\sf high}$) have von Neumann entropy $\frac{\log d}{2}$ (resp.~$\log d$) across the cut.

Fixing an edge $e \in T$, for each ensemble $\caD \in \{ \caD_{\sf low}, \caD_{\sf high} \}$, we construct a distribution of holographic states by ``inserting'' a new tensor $T_\psi$ for $\psi \sim \caD$ on the edge $e$ of $T$. Here, we view the $2$-qudit state $\psi$ as a $2$-leg tensor $T_\psi$ with local dimension $d$. Let $\caT_{\sf low}$ and $\caT_{\sf high}$ be the corresponding distributions of tensor network states. 

By a standard reduction, we can show that if a polynomial-time quantum algorithms can distinguish between $\caT_{\sf low}$ and $\caT_{\sf high}$, then it can also distinguish between $\caD_{\sf low}$ and $\caD_{\sf high}$. Since the latter two are indistinguishable, $\caT_{\sf low}$ and $\caT_{\sf high}$ are indistinguishable as well.

Let $T_{\sf high}$ (resp.~$T_{\sf low}$ ) be the weighted version of tree $T$ with each edge having weight $\ln d$ (resp.~$\frac{1}{2} \ln d$).
To show that states from $\caT_{\sf high}$ (resp.~$\caT_{\sf low}$) satisfy the RT formula exactly, we need to show that for every bipartition of the $n$ leaves (boundary states) into set $S$ and $[n] \setminus S$ the entanglement entropy of this cut is given by the RT formula. For the purpose of proof, we can imagine that we have also inserted new tensors representing the maximal mixed state $\sum_{i \in [D]} \spz{i}\spz{i}$ on every edge of $T$ except $e$, on which we already inserted a tensor drawn from the pseudoentangled state ensemble. This does not change the ensembles $\caT_{\sf low}$ and $\caT_{\sf high}$.  
We then follow a similar argument from~\cite{Pastawski_2015} based on max-flow min-cut theorem, and show that by cleverly constructing a path-covering of the tree (see \cref{sec:pe_link-tree-RT}), for any state from either $\caD_{\sf low}$ or $\caD_{\sf high}$, we can take the bi-partite states sitting on the min-cut between $S$ and $[n] \setminus S$ as input, and convert them into the whole state without changing the entropy across the cut. This proves the RT formula because the entropy on the min-cut is exactly the minimum cut between $S$ and $[n] \setminus S$ on $T_{\sf high}$ or $T_{\sf low}$.

\subsection{Discussion \label{sec:discussion}}

In this work we have constructed examples of states that satisfy the RT formula for radically different hyperbolic geometries, but which are computationally hard to distinguish from one another. 
There is no clear need for horizons in our construction, so this opens the possibility that geometry reconstruction might be exponentially intractable even outside of event horizons, due to the generic relation between entanglement and geometry. 
Our work does not settle this question, and more work needs to be done to make our construction more physically relevant.
First, our construction produces states which would be high energy in the CFT, as it does not take into account the CFT Hamiltonian. 
Applying a pseudorandom unitary to our states would necessarily boost their energy to something close to the Haar average.
Generically this might result in the formation of a black hole if the state is time evolved, as that ring of energy in the AdS could collapse into a black hole.
A natural open question is if our pseudoentanglement construction can be improved so that the pseudoentangled boundary state is also low-energy with respect to a given Hamiltonian.
One possible way to do this would be to make a version of a pseudoentangling pseudorandom unitaries which preserves the low energy subspace of a Hamiltonian:
\begin{question}
    Are there low-energy pseudoentangling PRUs? That is, given a local Hamiltonian $H$ and an energy cutoff $E$, does there exist a PRU construction that maps low-energy states to other low-energy states without dramatically altering the entanglement structure of the input state?
    Here, the security requirement of the PRU needs to be weakened so that it is only required to look Haar random within the low-energy subspace.
\end{question}
If such a PRU is possible, it would immediately create low-energy pseudoentangled CFT states, which would then not create horizons in AdS/CFT.
This would seriously question whether the ``quantum gravity in the lab'' paradigm could shed light on quantum gravity, even in situations without horizons. 
It would also open the question of what characterizes the hardness of geometry reconstruction.
In operator construction, exponential complexity is characterized by geometrical features of the bulk.
What is the analogous condition for geometry reconstruction?

Additionally, there is the question of how to interpret our construction from the Python's lunch perspective; we discuss this in detail in \cref{sec:pythons lunch}.
This turns out to be quite subtle, again due to incomparable input/output models which make translating a result from one setting to another tricky.
A related issue is whether the key in our pseudoentanglement constructions is public or private. 
As well as being pertinent to the Python's lunch discussion (see \cref{sec:pythons lunch} for details), the reliance of a holographic pseudoentanglement construction on a private key would weaken the link with AdS/CFT as it is typically assumed that the holographic map is completely known.
Our holographic pseudoentanglement construction from pseudoentangled link states can be instantiated using public key pseudoentanglement constructions (see \cref{sec:pe_link-planar-public-key}). 
However, the holographic pseudoentanglement construction from low-entangling PRUs  currently requires a private key.
Constructing a public key holographic pseudoentanglement scheme which can hide large differences in bulk geometry would strengthen the link with AdS/CFT and shed more light on the relationship between this work and the Python's lunch conjecture:
\begin{question}
    Can one create public-key holographic pseudoentanglement which hides large differences in the bulk geometry? 
    \end{question}

On the cryptography side, our argument based on FHE is restricted to a single copy because e.g.~the FHE scheme \cite{Mahadev:2017bwz} relies on the quantum one-time pad, which is a unitary one-design. 
This is necessary if one is considering the strongest form of security where \emph{no} properties of the original quantum state are accessible to observers who have multiple copies of the encoded state and the ability to apply \emph{arbitrary} operators to it homomorphically.
This follows because if one is given two copies of a state $V_k\ket{\psi}$, one can, for example, estimate expectation values of a binary operator $O_B$ via performing the operator homomorphically on one copy as $\tilde{O}_BV_k\ket{\psi} = V_k O_B\ket{\psi}$ and applying the SWAP test with the other copy of $V_k\ket{\psi}$.
Therefore it is not possible to hide every property of the encoded state in a multi-copy QFHE scheme that still allows homomorphic evaluation of \emph{all} operators. 
However, AdS/CFT does not exhibit the strongest form of homomorphic encryption because operator reconstruction depends on the bulk geometry, and hence on properties of the state that is being encoded (see \cref{sec:pythons lunch} for a discussion of how boundary duals of bulk operators are reconstructed and the dependence on the bulk geometry).
This raises a question about the existence of a variant of a quantum FHE scheme where the encryption is not ``fully'' homomorphic, but instead the homomorphic evaluation depends on some coarse-grained properties of the state:
\begin{question}
    Does there exist a quantum homomorphic encryption scheme which is multicopy secure for hiding some properties of the state, but where the homomorphic evaluation of operators can depend arbitrarily on some coarse-grained properties of the encrypted state?
\end{question}

A more open ended avenue for future research is to explore the relationship between AdS/CFT and FHE. 
It has previously been argued that AdS/CFT should be viewed as a quantum error correcting code (QECC)~\cite{Almheiri_2015}.
Could AdS/CFT also be an example of a (weakened kind of) FHE scheme?
Somewhat provocatively:
\begin{question}
    Does AdS/CFT=FHE?
\end{question}
In fact, Gottesman has recently discussed an intriguing conceptual connection between FHE and black holes \cite{aaronsonblog2022}, and our construction shows the relationship might be made more direct.
Indeed there exist examples of combined QECC and homomorphic encryption schemes~\cite{ouyang2022generalframeworkcompositionquantum,sohn2024errorcorrectableefficientquantum}, which strengthens the possibility.

\paragraph{Acknowledgements.}
We thank Netta Engelhardt, Andru Gheorghiu, Patrick Hayden, Henry Lin, Geoff Pennington, Renato Renner, Arvin Shahbazi-Moghaddam, Leonard Susskind, Douglas Stanford, and Lisa Yang for helpful discussions.
This work was done in part while the authors were visiting the Simons Institute for the Theory of Computing, supported by DOE QSA grant \#FP00010905.
C.A.~was supported by the Heising-Simons Foundation via Grant 2024-4848.
A.B.~and T.K.~were supported in part by the DOE QuantISED grant DE-SC0020360.  
A.B.~was supported in part by the U.S. DOE Office of Science under Award Number DE-SC0020377 and by the AFOSR under grants FA9550-21-1-0392 and FA9550-24-1-0089.
T.K.~is supported in part by SLAC (Q-NEXT) and QFARM. 
L.C.~is supported by a Miller Research Fellowship.
T.M.~acknowledges support from SNSF Grant No.~20CH21\_218782 and the ETH Zurich Quantum Center.
U.V.~was supported in part by DOE NQISRC QSA grant FP00010905, NSF QLCI Grant No.~2016245, and MURI Grant FA9550-18-1-0161.

\paragraph{Independent concurrent work.}
We note that independent work of Cheng, Feng and Ippoliti ~\cite{matteoPaper} and Engelhardt et al.~\cite{engelhardt2024spoofingentanglementholography} have also obtained pseudoentanglement constructions via different techniques. In particular~\cite{matteoPaper} produces a construction based on standard cryptographic assumptions using tensor network states, and~\cite{engelhardt2024spoofingentanglementholography} produces a construction based on the heuristic assumptions of~\cite{bouland2019computationalpseudorandomnesswormholegrowth}. We view these as complementary to our results.

\section{Preliminaries}

In this section, we introduce the necessary preliminaries for this work. We begin by introducing some notation. 

\paragraph*{Notation.} We use $\N$ and $\posN$ to denote the set of all non-negative integers and the set of all positive integers, respectively. For a set $S$, we use $\Haar(S)$ to denote the Haar measure on it. 

\newcommand{\rmS}{\mathrm{S}}
\newcommand{\rmD}{\mathrm{D}}
\newcommand{\rmU}{\mathrm{U}}
\newcommand{\rmL}{\mathrm{L}}

For a Hilbert space $\H$, we use $\rmS(\H)$, $\rmD(\H)$, $\rmU(\H)$, and $\rmL(\H)$ to denote the set of pure quantum states, density operators, unitary operators and bounded linear operators on $\H$, respectively. For a (continuous) distribution $\caD$, we often write $\Ex_{x \sim \caD} [f(x)]$ to denote $\int_{x \sim \caD} f(x) dx$ for brevity.

For a unitary $U \in \rmU(\H)$, we use $U^{\dagger t}$ to denote $\left(U^\dagger\right)^{\otimes t}$ for simplicity.

\paragraph*{Graphs and min-cuts.} For a weighted graph $G = (V,E)$ and two disjoints sets $A,B \subseteq V$, we write $\mincut_{A,B}(G)$ as the min-cut between $A$ and $B$ in $G$. 

Formally, we say a set of edges $\gamma \subseteq E$ is a \emph{cut} between $A$ and $B$ on $G$ if $A$ and $B$ are disconnected in $G$ after removing $\gamma$. Slightly abusing the notation, we also say a set $A \subseteq S \subseteq (V \setminus B)$ is a \emph{cut} between $A$ and $B$ on $G$. More formally, for every set $S \subseteq V$, we define $\Weight_G(S)$ as the total weights of edges with exactly one end-points contained in $S$, and we set
\[
\mincut_{A,B}(G) = \min_{A \subseteq S \subseteq (V \setminus B)} \Weight_G(S).
\]

\subsection{Pseudorandomness and pseudoentanglement}\label{sec:pseudoentanglement}

Now we define the central notion of this work: pseudoentangled holographic state ensembles (PES).

\newcommand{\caK}{\mathcal{K}}
\newcommand{\sfK}{\mathsf{K}}

\paragraph*{Pseudoentangled holographic states.} In~\cite{aaronson2023quantumpseudoentanglement} a pseudoentangled state ensemble (PES) is defined as two ensembles of quantum states that are (1) computationally indistinguishable and (2) with high probability, a state drawn from one ensemble has a high entropy across every cut and a state drawn from another one has low entropy across every cut. In this work, we will consider states that (approximately) satisfy the RT formula, which motivates the following definition.

For a weighted graph $G$ with at least $n$ vertices, we say an $n$-qudit quantum state $\rho$ \emph{has holographic entropy structure $G$} (or, it satisfies the RT-formula with respect to a weighted graph $G$), if for every $A \subseteq [n]$, it holds that\footnote{Here, we should think of the vertices $1,\dotsc,n$ corresponds to the boundary nodes, and the rest of vertices correspond to the bulk node.}
\[
S_{A}(\rho) = \mincut_{A,[n] \setminus A}(G),
\]
where $ \mincut_{A,[n] \setminus A}(G)$ denotes the minimum cut between $A$ and $[n] \setminus A$ in $G$. We also say a distribution $\caD$ of $n$-qudit quantum states have holographic entropy structure $G$ if all $\rho$ in the support of $\caD$ has holographic entropy structure $G$.

We can also relax the requirement for an ensemble of quantum states by only asking the entropy to be approximated by $\mincut_{A,[n] \setminus A}(G)$ with high probability. 

Formally, we say a distribution $\caD$ of $n$-qudit quantum states \emph{has holographic entropy structure approximated by $G$} (or, it satisfies the RT-formula approximately with respect to a weighted graph $G$), if for every $A \subseteq [n]$ and every $\tau \in (0,1)$ it holds that with $1-\tau$ probability over $\rho \sim \caD$, 
\[
S_A(\rho) \ge \mincut_{A,[n] \setminus A}(G) \cdot (1 - o(1)) - \ln\tau^{-1}\mathdot\footnote{
	In particular, if $\mincut_{A,[n] \setminus A}(G) \ge \omega(\ln n)$ (in our work, each edge always has weight $\ln q > \ln n$, so this means the cut has super-constant edges), then we can pick $\tau = n^{-\omega(1)}$ and it follows that for every $A \subseteq [n]$ it holds that with $1-n^{-\omega(1)}$ probability over $\rho \sim \caG$ that $S_A(\rho) \ge \mincut_{A,[n] \setminus A}(G) \cdot (1 - o(1))$.}
\]

Then, we are ready to state our more general definition of pseudoentangled state ensemble (PES) with holographic entropy structure $G$ vs $H$.

\begin{definition}[Pseudoentangled holographic states with entropy structure $G$ vs $H$]
	Let $\lambda$ be the security parameter. Let $\H = \{ \H_{\lambda} \}_{\lambda \in \posN}$ and $\sfK = \{ \sfK_{\lambda} \}_{\lambda \in \posN}$ be a family of Hilbert spaces and a family of key spaces. Let $G = \{G_\lambda\}_{\lambda \in \posN}$ and $H = \{H_\lambda\}_{\lambda \in \posN}$ be two families of weighted graphs. Two keyed families of quantum states $\{\spz{\Phi}_{k} \in \rmS(\H)\}_{k \in \sfK}$ and $\{\spz{\Psi}_{k} \in \rmS(\H)\}_{k \in \sfK}$ (parameterized by $\lambda$) form a pseudoentangled holographic state ensemble (PES) with exact (resp.~approximate) entropy structure $G$ vs $H$, if the following three conditions hold:
	
	\begin{enumerate}
		\item There is a polynomial-time quantum algorithm $G_{\Phi}$ (resp.~$G_{\Psi}$) that generates state $\spz{\Phi_k}$ (resp.~$\spz{\Psi_k}$) on input $k \in \sfK$.
		
		\item The state ensemble $\{\spz{\Phi}_{k} \in \rmS(\H)\}_{k \in \sfK}$ has entropy structure (resp.~approximated by) $G$, and the state ensemble $\{\spz{\Psi}_{k} \in \rmS(\H)\}_{k \in \sfK}$ has entropy structure (resp.~approximated by) $H$.
		
		\item For any polynomial $m \le \poly(\lambda)$, and any polynomial-time quantum algorithm $A$, it holds that
		\[
		\left| \Pr_{k \leftarrow \sfK_{\lambda}}\left[ A\left(\spz{\Psi_k}^{\otimes m}\right) = 1\right] - \Pr_{k \leftarrow \sfK_{\lambda}}\left[ A\left(\spz{\Phi_k}^{\otimes m}\right) = 1\right] \right| \le \negl(\lambda)\mathcomma
		\]
	\end{enumerate}

\end{definition}

\subsection{Holographic quantum error correcting codes}

Holographic quantume error correcting codes are toy models of AdS/CFT built out of tensor networks.

\begin{definition}[Holographic quantum error correcting code (HQECC), modified from~\cite{Pastawski_2015}] 
Consider a tensor network which is embedded in a tessellation of $\mathbb{H}^2$ by some Coxeter polytope. The tensor network is called a holographic quantum error correcting code if it gives rise to an isometric map from uncontracted bulk legs to uncontracted boundary legs.
\end{definition}

There are numerous examples of HQECC built out of perfect~\cite{Pastawski_2015}, random~\cite{Hayden_2016}
 and random stabilizer~\cite{Apel_2022}.
 In the remainder of this section we define each type of tensor and collect key facts about them.
 
\subsubsection{Perfect Tensors}\label{sec:prelim:perfect-tensors}

For an even $n$, an $n$-index tensor $T_{a_1,\dotsc,a_{n}}$ is a \textbf{perfect tensor} if, for any bipartition of its indices into a set $A$ of size $n/2$ and its complement $A^c$, $T$ is a unitary transformation from $A$ to $A^c$ (after normalization). 

\subsubsection{Random tensors}
\label{Random tensors}

Random tensors can be generated via random states on the respective Hilbert space. To obtain the random state $\ket{\phi} = U\ket{0}$, start from an arbitrary reference state, $\ket{0}$, and apply a random unitary operation, $U$. The average over a function of the random state, $f(\ket{\phi})$, is given by integration over the unitary group, $U$, with respect to the Haar measure
\begin{equation}
	\langle f(\ket{\phi}) \rangle = \int_{\mathcal{U}(d)} f(\ket{\phi}) dU.
\end{equation}
The Haar probability measure is a non-zero measure $\mu$ such that if $h$ is a probability density function on the group $G$, for all $S\subseteq G$ and $g\in G$:
\begin{equation}
	\mu(gS) = \mu(Sg) = \mu(S),
\end{equation}
where
\begin{equation}
	\mu(S)  := \int_{g\in S} d\mu (g) = \int_{g\in S} h(g) dg, \qquad \mu(G)  := 1.
\end{equation}
A unique Haar measure exists on every compact topological group, in particular the unitary group.

\subsubsection{Random stabilizer tensors}
\label{Random stabilizer tensors}

Random \textit{stabilizer} tensors are analogously generated by uniformly choosing \textit{stabilizer} states at random. In this case the reference state is chosen as a stabilizer state $\ket{\tilde{\psi}}$, stabilized by $S$, and instead of a random unitary, a random Clifford unitary, $C$, is applied to generate the random stabilizer state $\ket{\psi} =C \ket{\tilde{\psi}} $. Since elements of the Clifford group map the Pauli group to itself under conjugation, the resulting state is stabilized by $S' = C S C^\dagger$:
\begin{subequations}
	\begin{align}
		C P C^\dagger \ket{\psi} & = C P C^\dagger C \ket{\tilde{\psi}}\\
		& =  C P \ket{\tilde{\psi}}\\
		& =  C \ket{\tilde{\psi}}\\
		& = \ket{\psi}
	\end{align}
\end{subequations}
In the case of qudits of prime dimension the same procedure is followed for generating random stabilizer tensors, substituting for the generalised Pauli and Clifford operators.

\begin{theorem}[Random stabilizer tensors are perfect~\cite{Apel_2022}]
	\label{thm Random stabilizer tensors are perfect}
	Let the tensor $T$, with $t$ legs, describe a stabilizer state $\ket{\psi}$ chosen uniformly at random where each leg corresponds to a prime $p$-dimensional qudit. The tensor $T$ is perfect with probability
	\begin{equation} \label{high_prob_eqn}
		P \geq  \max \left\{ 0, 1 - \frac{1}{2p^b}{t\choose \floor{t/2}}\right\}
	\end{equation}
	in the limit where $p$ is large, where $0< b \leq 1$.
\end{theorem}

\section{Holographic maps with homomorphic encryption}\label{sec:FHE_holo_maps}

Homomorphic encryption demonstrates that the complexities of geometry reconstruction and operator reconstruction can differ. 
Those schemes, however, involve a secret key that is unknown to the person trying to reconstruct the logical information. 
It is unclear what role such a secret key has to play in the bulk-to-boundary maps of AdS/CFT.
If the AdS to CFT map is one fixed linear map that applies equally well to all geometries, then in principle nothing stops us from knowing that map fully, and there should be no analog of a secret key.
In this section we argue that even in the strongest setting possible, where we assume the bulk-to-boundary map is fully known and there is no secret key, we can still demonstrate a gap between the complexity of operator reconstruction and state reconstruction.

The most obvious way of achieving this is to apply the homomorphic encryption scheme and then simply trace out the key. The downside of this is that now the bulk-to-boundary map loses information -- we have merely shifted lack of knowledge of the map into lack of knowledge of the traced-out system. 
We can circumvent this problem by effectively placing the secret key inside a Python's lunch. 
This creates a situation where the key is not information-theoretically lost but accessing it is computationally intractable due to the Python's lunch obstruction. 
As a result, geometry construction remains intractable, whereas operator reconstruction never relied on knowledge of the key or properties in the first place and remains easy.
 
We can formalize this idea using a post-selection-based argument similar to~\cite{brown2019pythonslunchgeometricobstructions}.
Consider a QFHE scheme, which is a collection of isometries $\{V_x\}_x$ indexed by keys $x$.
Assume that the QFHE scheme is information-theoretically decryptable (i.e., all the images of the different isometries $V_x$ are orthogonal for different keys $x$).
Let $\mathcal{H}_\mathrm{key}, \mathcal{H}_b,$ and $\mathcal{H}_B$ be finite dimensional Hilbert spaces.
By running the scheme coherently on a key register, this gives rise to an  encryption map 
\begin{align*}
V: \mathcal{H}_\mathrm{key} \otimes \mathcal{H}_b \to \mathcal{H}_\mathrm{key} \otimes \mathcal{H}_B
\end{align*}
and a decryption map 
\begin{align*}
V_{\rm dec}: \mathcal{H}_\mathrm{key} \otimes \mathcal{H}_B \to \mathcal{H}_\mathrm{key} \otimes \mathcal{H}_b
\end{align*}
with the following properties:
\begin{enumerate}
	\item It is easy to implement the encryption map $V$:
	\begin{equation}
		V \ket{x}_\mathrm{key}\ket{\psi}_b  = \ket{x}_\mathrm{key} V_x \ket{\psi}_b~.
	\end{equation}
	\item It is hard to implement the decryption map $V_\mathrm{dec}$:
	\begin{equation}
		V_\mathrm{dec} \ket{0}_\mathrm{key} V_x \ket{\psi}_b = \ket{x}_\mathrm{key} \ket{\psi}_b~.
	\end{equation}
	\item For any unitary $U$ on $\mathcal{H}_b$, it is easy to implement a unitary $\widetilde{U}$ such that for any $x$
	\begin{equation}
		\widetilde{U} V_x \ket{\psi}_b = V_x U \ket{\psi}_b~.
	\end{equation}
\end{enumerate}
It might be surprising that conditions (i) and (ii) can coexist.
After all, if $V$ is easy to implement then inverting it on its image is also easy.
The argument goes like this: we can in general represent the isometry $V$ as 
\begin{equation}
	V = W \ket{0}_A~,
\end{equation}
where $A$ is some arbitrary ancilla system and $W$ is a unitary on $\mathcal{H}_\mathrm{key} \otimes \mathcal{H}_\mathrm{b} \otimes \mathcal{H}_\mathrm{A}$. 
If $V$ is easy to implement, that implies $W$ is an efficient unitary operator. 
But then $W^\dagger$ is also an efficient unitary operator.
Therefore, given some state $V \ket{x}\ket{\psi}$, we can easily undo $V$ by interpreting this as
\begin{equation}
    V \ket{x}\ket{\psi} = W \ket{x}\ket{\psi}\ket{0}_A
\end{equation}
and acting with $W^\dagger$ and then measuring $A$, which will have outcome $\ket{0}_A$ with probability $1$.
Crucially, however, this operation need not do anything nice when acted on $\ket{0} V_x \ket{\psi}$!

Note that it is important that $V_\mathrm{dec}$ is required to work for \emph{any} $x$ in condition (ii), because $V_x$ is itself easy to implement (because $V$ is), so by an argument similar to above, $V_x$ is also easy to invert on its image, i.e. states of the form $V_x \ket{\psi}_b$.
But that's not enough to give us an efficient $V_\mathrm{dec}$ that works for \emph{all} $x$.  

These conditions ensure that given $V_x \ket{\psi}_b$ for \emph{unknown} $x$, state reconstruction is hard but operator reconstruction is easy.
However, this is not yet what we want if we want the holographic map to be a \emph{completely known} linear map.
In that view, we cannot say ``the holographic map is a $V_x$ for some unknown $x$''. 
We want some particular linear map that nonetheless has a gap in operator and state reconstruction. 

This can be constructed as follows.
Let the ``bulk'' Hilbert space be 
\begin{equation}
	\mathcal{H}_\mathrm{bulk} = \mathcal{H}_{\mathrm{key}} \otimes \mathcal{H}_b~,
\end{equation}
and let $V: \mathcal{H}_\mathrm{bulk} \to \mathcal{H}_B$ be a QFHE scheme.
Let $Q: \mathcal{H}_\mathrm{bulk} \to \mathcal{H}_B$ be defined as
\begin{equation}
    Q = \sum_{x} \bra{x}_\mathrm{key} V~.
\end{equation}
Then it is straightforward to see that $Q \ket{x}_\mathrm{key} \ket{\psi}_b = V_x \ket{\psi}_b$.
Note that it follows from the existence of the decryption map (or more precisely the orthogonality of the images of the different $V_x$) that $Q$ preserves the normalization of input states. 
This $Q$ is the bulk-to-boundary map with the properties we want.
To see this, let $\mathcal{H}_R$ be an arbitrary reference system, and consider an arbitrary state
\begin{equation}
	\ket{\phi}_\mathrm{bulk,R} = \sum_{x,y,z} c_{xyz} \ket{x}_{\mathrm{key}} \ket{y}_{b} \ket{z}_R~,
\end{equation}
where $x,y,z$ are labels for orthonormal bases.
Acting our map, we obtain the state
\begin{equation}
	Q \ket{\phi} = \sum_{x,y,z} c_{xyz} \left(V_x \ket{y}_{b}\right)\ket{z}_R~.
\end{equation}
Now imagine we have access to just $B$, not $R$. 
In principle we can recover the bulk state $\ket{\phi}$ by using $V_\mathrm{dec}$, but in general that won't be easy.
State reconstruction is hard!
On the other hand, reconstructing operators is easy, for any operator on $\mathcal{H}_b$.
Indeed by condition (iii), for any $U_b$ there exists an efficient $U_B$ such that 
\begin{equation}
	U_B Q \ket{\phi} = Q U_b \ket{\phi}~.
\end{equation}
Note that reconstructing operators that act on $\mathcal{H}_\mathrm{key}$ is not necessarily easy.
Furthermore this $Q$ is itself not necessarily easy to implement, because its definition involves postselection. 
Still, this demonstrates that there exists a linear map such that state reconstruction is hard but operator reconstruction is easy for many operators.

\section{Holographic pseudoentanglement from low-entangling PRUs} \label{sec:pe_pru}
\subsection{Pseudorandom unitaries}\label{sec:prus}

This construction relies on the shallow depth PRUs of~\cite{brickwork_t_design}.

\begin{definition}[Pseudorandom unitary\cite{metger2024simpleconstructionslineardepthtdesigns}]
Let $n \in \mathbb{N}$ be the security parameter. An infinite sequence $\mathcal{U} = \{\mathcal{U}_{n \in \mathbb{N}}$ of $n$-qubit unitary ensembles $\mathcal{U}_n = \{U_k\}_{k \in \mathcal{K}}$ is a pseudorandom unitary if it satisfies the following conditions:

\begin{itemize}
    \item (Efficient computation) There exists a polynomial-time quantum algorithm $\mathcal{Q}$ such that for all keys $k \in \mathcal{K}$ where $\mathcal{K}$ denotes the key space, and any $\ket{\psi} \in (\mathbb{C}^2)^{\otimes n}$ it holds that
    \begin{equation}
    \mathcal{Q}(k \ket{\psi}) = U_k \ket{\psi}
    \end{equation}
    \item (Pseudorandomness) The unitary $U_k$ for a random key $k \sim \mathcal{K}$ is computationally indistinguishable from a Haar random unitary $U \sim \textrm{Haar}(2^n)$. In other words, for any quantum polynomial-time (QPT) algorithm $\mathcal{A}$ it holds that
    \begin{equation}
    |\mathrm{Pr}_{k \sim \mathcal{K}}[\mathcal{A}^{U_k}(1^\lambda)=1]- \mathrm{Pr}_{U\sim\mathrm{Haar}}[\mathcal{A}^{U_k}(1^\lambda)=1]| \leq \mathrm{negl}(n)
    \end{equation}
\end{itemize}
\end{definition}

Initial constructions of PRUs required circuit depths polynomial in $n$~\cite{metger2024simpleconstructionslineardepthtdesigns,chen2024efficientunitarydesignspseudorandom}.
In~\cite{brickwork_t_design} this was improved via a construction that requires circuit depth $\poly \log (n)$.
It is a `brickwork' construction, that builds a PRU on $n$ qubits by `patching together' PRUs on $\omega(\log n)$ qubits (see \cref{fig:brickwork PRU}).
Since the circuit is low-depth and local it cannot change the long range entanglement structure of the state it acts on, it can only create or destroy entanglement between nearest- and next-nearest-neighbour patches of $\omega(\log n)$ qubits.

\begin{figure}
\centering
\includegraphics[scale=0.6]{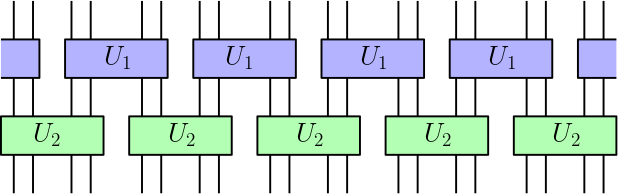}
\caption{The brickwork PRU construction from~\cite{brickwork_t_design}. Each unitary acts on $\omega(\log n )$ qubits. Two layers of these small PRUs generates a PRU on $n$ qubits.}\label{fig:brickwork PRU}
\end{figure}

\subsection{Construction}
We will take two different bulk geometries and use the shallow-depth, low-entangling PRUs to `hide' the difference in geometries, so that any observer with access to polynomially many copies of the boundary state cannot distinguish the two states.
In order to present a concrete construction we start from HQECC, however we note that the idea of applying shallow depth PRUs to `hide' geometry could equally well be applied to CFT states arising in the boundary of AdS/CFT.
The concrete construction has three simple steps:

\paragraph{Step 1:}
Take two arbitrary geometries, $g_1,g_2$ and cut them off at some finite radius.
Tessellate them both in such a way that both geometries have $n$ edges on the boundary of the tessellation for some $n \in \mathbb{N}$. 
Note that we do not require that the geometries are cut off at the same radius, just that the number of boundary edges is the same in both cases.
Crucially we assume that the two geometries are substantially different.
Let $l_i$ denote the number of faces in the polygon that tessellates $g_i$.

\paragraph{Step 2:} 
Construct a HQECC for the tessellations of $g_1,g_2$ using random stabilizer tensors.\footnote{We use random stabilizer tensors so that our HQECC exactly obeys the RT formaula and can be efficiently instantiated on a quantum computer.}
We will choose $l_i+1$-index tensors where each tensor leg has dimension $D = 2^{\omega(\log n)} = n^q$ for $q>1$.
Note that with these choices with high probability the tensor networks are isometries from bulk to boundary which exactly obey the RT formula, i.e.~for a boundary region $A$ the entanglement entropy of $\rho_A$ satisfies:
\begin{equation}
S(\rho_A) = \log D \cdot |\gamma_A| = q \log n \cdot |\gamma_A|
\end{equation}
where $\gamma_A$ is the minimal geodesic in the tessellation that has its end points at the boundary of $A$.
This follows from the results of~\cite{Apel_2022}.

\paragraph{Step 3:} Apply the brickwork PRU construction from~\cite{brickwork_t_design} to the boundary states of each tensor network (see \cref{fig:hqecc pru} for an illustration).
As outlined in the previous section, the brickwork construction relies on implementing PRUs on $\omega(\log n)$ qubits to construct a PRU on $n$ qubits.
Existing constructions of PRUs are non-adaptive, therefore we will need to use two different PRU ensembles -- one for each layer of the brickwork -- to achieve security.\footnote{A PRU is said to be non-adaptive if it is only secure against algorithms which query the PRU in parallel. An adaptive PRU on the other hand would be secure against algorithms which can make sequential queries to the PRU.}

\begin{figure}
\centering
\includegraphics[scale = 0.3]{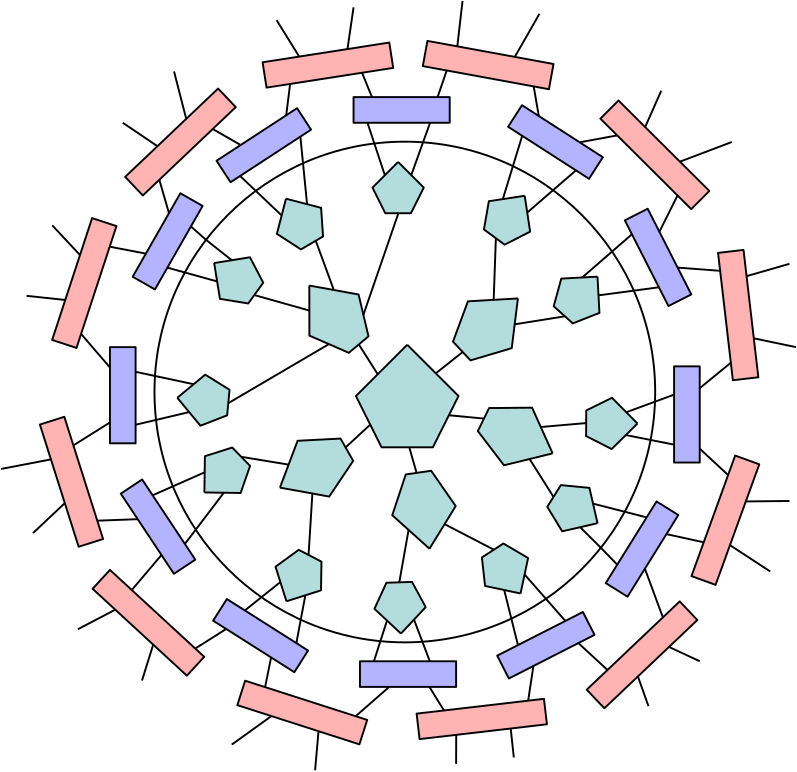}
\caption{We apply the brickwork PRU construction to the boundary of a HQECC to `hide' the geometry.}\label{fig:hqecc pru}
\end{figure}

The brickwork PRU construction is low-entangling because it has a small lightcone, however it will slightly modify the entanglement structure of the boundary states.
The entanglement corresponding to any  boundary region $A$ can now be bounded by:\footnote{This follows because the brickwork PRU can entangle or disentangle at most a constant number of outputs from the HQECC.}
\begin{equation}
S(\rho_A) = q \log n \cdot |\gamma_A| \pm O(q \log n)  = q \log n \cdot \left(|\gamma_A| + O(1) \right)
\end{equation}

This construction leads to an example of hardness of geometry reconstruction, which we capture in the following theorem: 

\begin{theorem}\label{thm:PRU hqecc} For any two bulk geometries $g_1,g_2$, there exist two ensembles of quantum states $\ket{\Psi}_k$,$\ket{\Phi}_k$ such that
\begin{enumerate}
    \item The states $\ket{\Psi}_k$,$\ket{\Phi}_k$ are efficiently constructable by poly-size quantum circuits.
    \item The entanglement entropy of the states $\ket{\Psi}_k$,$\ket{\Phi}_k$ (for any choice of key $k$) satisfies:
    \begin{equation}
S(\rho_A) = q \log n \cdot |\gamma_A| \pm O(q \log n)  = q \log n \cdot \left(|\gamma_A| + O(1) \right)
    \end{equation}
    where $A$ is some boundary region, $\rho_A$ is the reduced density matrix of $\ket{\Psi}_k$ (resp.~$\ket{\Phi}_k$) on $A$ and $|\gamma_A|$ is the length of the minimal cut through $g_1$  (resp.~$g_2$) that shares a boundary with $A$
    \item No poly-time quantum algorithm can distinguish a random $\ket{\Psi}_k$ from a random $\ket{\Phi}_k$ given polynomially many copies of the state.
\end{enumerate}
\end{theorem}
\begin{proof}

The states are efficiently constructable because the boundary state (before applying the PRU) is obtained by acting on the all zero state with a random Clifford circuit, and random Clifford circuits can be efficiently implemented.
The PRUs themselves are efficiently implementable by the arguments of~\cite{brickwork_t_design}.

The entanglement scaling follows because the brickwork PRU can only modify the short range entanglement of the state -- it can create and destroy entanglement only between a constant number of tensor legs, and the tensor legs are composed of $q \log n $ qubits.

By the arguments of~\cite{brickwork_t_design} these boundary states are both indistinguishable from Haar random states to any poly-time bounded observer with access to polynomially many copies of the state.
This follows because the brickwork PRU construction is indistinguishable from a Haar random unitary, and the result of applying a Haar random unitary to an arbitrary state is a Haar random state.
Since the boundary states are both indistinguishable from Haar random states, they are also indistinguishable from each other.
\end{proof}

\cref{thm:PRU hqecc} applies to arbitrary bulk geometries, so in particular we can apply it to two geometries which are easy to distinguish in the bulk gravitational theory. 
The fact that the two geometries can be efficiently distinguished but (polynomially many copies of) the boundary state cannot be distinguished implies that the geometry cannot be reconstructed given access to just the boundary state.
Note that this argument applies to arbitrary bulk geometries $g_1,g_2$, therefore in particular it does not rely on the existence of a horizon in either geometry.

\section{Holographic pseudoentanglement from pseudoentangled link states} \label{sec:pe_link}

In this section, we present our constructions of pseudoentangled holographic states via pseudoentangled link states. In~\cref{sec:pe_link-prelim}, we first introduce the necessary preliminaries. Then, in~\cref{sec:pe_link-tree}, we present our construction based on tree tensor networks, which satisfies the exact RT formula. Finally, in~\cref{sec:pe_link-planar}, we present our construction based on random stabilizer tensor networks that only satisfies RT formula approximately. Both of our constructions can be made public-key pseudoentangled states.

\subsection{Preliminaries} \label{sec:pe_link-prelim}

We will need the following construction of pseudoentangled states from~\cite{aaronson2023quantumpseudoentanglement}.

\begin{theorem}[{\cite[Theorem 2.5]{aaronson2023quantumpseudoentanglement}}]\label{theo:pseudo-ent}
	Let $D \in \posN$ be such that $\log D =\omega(\log n)$, and $k \in \posN$ be such that $k \le D$ and $\log k = \omega(\log n)$ (both $D$ and $k$ are parameterized by $n$). Let $\bF\colon [D] \to \{0,1\}$ be a quantum-secure pseudorandom function and $\bP \colon [D] \to [D]$ be a quantum-secure pseudorandom permutation (both against $\poly(n)$-time quantum adversaries). The following two distributions over quantum states
	\[
	\frac{1}{\sqrt{D}} \sum_{i \in [D]} (-1)^{f(i)} \spz{i} \quad \text{where } f \sim \bF
	\]
	and
	\[
	\frac{1}{\sqrt{|S|}} \sum_{i \in S} (-1)^{f(i)} \spz{i} \quad \text{where } f \sim \bF, p \sim \bP, S = \{ p(i) : i \in [k] \}
	\]
	are computationally indistinguishable against $\poly(n)$-time quantum adversaries given $\poly(n)$ many copies. 
\end{theorem}

We remark that~\cite{aaronson2023quantumpseudoentanglement} proved that $\caD_{\sf subset}$ is a pseudorandom state ensemble,\footnote{see~\cite{JLS18} for a formal definition of pseudorandom state ensemble} and the previous work~\cite{BS19} proved that $\caD_{\sf full}$ is pseudorandom. The theorem above follows as a simple corollary of these two results.

The following corollary is straightforward from~\cref{theo:pseudo-ent}.

\begin{corollary}\label{cor:pseudo-ent}
	Let $D,k,\bF,\bP$ be the same as in~\cref{theo:pseudo-ent}. The following two distributions over quantum states
	\[
	\caD_{\sf full} \colon \frac{1}{\sqrt{D}} \sum_{i \in [D]} (-1)^{f(i)} \spz{i}\spz{i} \quad \text{where } f \sim \bF
	\]
	and
	\[
	\caD_{\sf subset} \colon\frac{1}{\sqrt{|S|}} \sum_{i \in S} (-1)^{f(i)} \spz{i}\spz{i} \quad \text{where } f \sim \bF, p \sim \bP, S = \{ p(i) : i \in [k] \}
	\]
	are computationally indistinguishable against $\poly(n)$-time quantum adversaries given $\poly(n)$ many copies.
\end{corollary}
\begin{proof}
	Note that applying unitary $\spz{x}\spz{y} \mapsto \spz{x}\spz{x+y}$ on the two ensembles of states from~\cref{theo:pseudo-ent} (with $\spz{0}$ padded at the end) gives the two ensembles of states in the corollary. Hence, if the two ensembles of states in the statements are computationally distinguishable, so are the two ensembles of states from~\cref{theo:pseudo-ent}.
\end{proof}

We will also need the following construction of perfect tensors from~\cite{helwig2012absolute,helwig2013absolutely}. 

Let $n \in \N$ be even and $D \in \N$ be such that $n \le D$ and $D$ is a prime power. We use $\F=\F_D$ to denote the finite field of size $D$ and $\omega_1,\dotsc,\omega_{n}$ to denote the first $n$ elements from $\F=\F_D$ (in an arbitrary but fixed ordering), and $\F_{<n/2}[X]$ denote the set of all polynomials in $\F[X]$ with degree less than $n/2$.

We have the following lemma, which follows from the argument from~\cite[Section~4.2]{helwig2013absolutely}; we include a self-contained proof for completeness.
\begin{lemma}\label{lm:perfect-tensor}
	Let $n \in \N$ be even and $D \in \N$ be such that $n \le D$ and $D$ is a prime power.
	\[
	\pt_{n,D} = \sum_{p \in \F_{<n/2}[X]} \bigotimes_{i \in [n]} \spz{p(\omega_i)}.
	\]
	is a perfect tensor with $n$ legs and bond dimension $D$. Moreover, for any bipartition of its indices into a set $A$ of size $n/2$ and its complement $A^c$, the corresponding unitary transformation from $A$ to $A^c$ has a $\poly(n,\log D)$-size quantum circuit.
\end{lemma}
\begin{proof}
	By symmetric, it suffices to show $\pt_{n,D}$ (interpreted as a tensor) is a unitary transformation from $[n/2]$ to $\{n/2+1,\dotsc,n\}$. Note that a polynomial $p \in \F_{<n/2}[X]$ is uniquely determined by its values on $\omega_1,\dotsc,\omega_{n/2}$, via interpolation. Hence, for $\alpha_1,\dotsc,\alpha_{n/2} \in \F$, we write $p_{\alpha_1,\dotsc,\alpha_{n/2}}$ to be the unique polynomial from $\F_{<n/2}[X]$ such that $p_{\alpha_1,\dotsc,\alpha_{n/2}}(\omega_i) = \alpha_i$. Also note that $p_{\alpha_1,\dotsc,\alpha_{n/2}}$ can be constructed from $\alpha_{1},\dotsc,\alpha_{n/2}$ in $\poly(n,\log D)$ time via standard interpolation.
	
	Now, we can write $\pt_{n,D}$ as
	\[
	\pt_{n,D} = \sum_{p \in \F_{<n/2}[X]} \bigotimes_{i \in [n]} \spz{p(\omega_i)} = \sum_{\alpha_1,\dotsc,\alpha_{n/2}} \bigotimes_{i \in [n/2]} \spz{\alpha_i} \otimes \bigotimes_{i \in \{n/1+1,\dotsc,n\}} \spz{p_{\alpha_1,\dotsc,\alpha_{n/2}}(\omega_i)}.
	\]
	
	Clearly, $\pt_{n,D}$ is a unitary transformation from $\bigotimes_{i \in [n/2]} \spz{\alpha_i}$ to $\bigotimes_{i \in \{n/1+1,\dotsc,n\}} \spz{p_{\alpha_1,\dotsc,\alpha_{n/2}}(\omega_i)}$, and this can be implemented by a $\poly(n,\log D)$-size quantum circuit.
\end{proof}

\subsection{Pseudoentangled holographic states via Tree Tensor Networks and Perfect Tensors} \label{sec:pe_link-tree}

We say a weighted tree is \emph{nice} if all intermediate (i.e., non-leaf) nodes have even degrees and all edge weights are the same. In this subsection, we give the following construction of pseudoentangled holographic states with exact entropy structure given by trees.

\begin{theorem}[Pseudoentangled holographic tree states]\label{theo:PSE-tree}
	Consider the following two graph families:
	\begin{itemize}
		\item Let $T$ be a nice tree with $n$ leaves and edge weight $\ln D \ge \omega(\ln n)$ ($T$ is parameterized by $n$). 
		\item Let $T_{[e]}$ be the tree that is the same as $T$ except that the weight of edge $e$ is reduced to $(\ln D) / 2$ from $\ln D$ (both $e$ and $T_{[e]}$ are parameterized $n$).
	\end{itemize}
	Assuming quantum-secure one-way functions exist, the following holds:
	\begin{itemize}
		\item There are two ensembles of quantum states $\caD_{T}$ and $\caD_{T,e}$ that constitute a pseudoentangled holographic state ensemble with exact entropy structure $T$ vs $T_{[e]}$.
	\end{itemize}
\end{theorem}

In the rest of this section, we will first define the ensembles $\caD_{T}$ and $\caD_{T,e}$ in~\cref{sec:pe_link-tree-ensembles}, and then prove~\cref{theo:PSE-pk-tree} in~\cref{sec:pe_link-tree-RT} and~\cref{sec:pe_link-tree-ind}. Finally, we generalize~\cref{theo:PSE-tree} to the public-key version in~\cref{sec:pe_link-tree-public-key}.

\subsubsection{The tensor network and the holographic state} \label{sec:pe_link-tree-ensembles}

We first specify the construction of quantum states $\caD_{T}$ and $\caD_{T,e}$. We will construct them using tree tensor networks consisting of perfect tensors from~\cref{lm:perfect-tensor}. Let $n,D \in \posN$ be such that $D > n$ and let $T$ be a nice tree with $n$ leaves (one can show that $n$ is even and all intermediate nodes have degree at most $n$) with all edge weights being $\ln D$. 

In the following, we first describe the standard procedure of turning a tree into a tree tensor network (state) and introduce some notation. Then, we show how to modify the resulting tree tensor network to obtain our construction.

Let $\TN_{T}$ be the tensor network obtained by replacing all intermediate nodes $u$ with the corresponding tensor $\pt_{d_u,D}$, where $d_u$ is the degree of $u$; see~\cref{fig:tree-tensor-network} for an illustration. We also let $\State(\TN_{T})$ be the $n$-qudit quantum state with qudit dimension $D$ that is obtained by contracting all intermediate tensor edges from $\TN_{T}$ to obtain a tensor with $n$ legs and bond dimension $D$, and then normalizing the resulting tensor. For simplicity, we often just use $\State(T)$ to denote $\State(\TN_T)$.

\begin{figure}[h]
	\begin{center}
		\begin{tikzpicture}[scale=0.8]
			\draw[dashed,line width=1.5] (0,0) circle (3cm);
			% Nodes
			\node[rectangle,draw,fill=white] (center1) at (-1,0) {$\pt_{4,D}$};
			\node[rectangle,draw,fill=white] (center2) at (1,0) {$\pt_{4,D}$};
			
			% Leaves
			\foreach \angle/\label in {60/1, 120/2, 180/3, 240/4, 300/5, 360/6} {
				\node[circle,draw,fill=white] (leaf\label) at (\angle:3cm) {\label};
			}
			
			% Connections
			\draw[line width=2, color=blue] (center1) -- (center2);
			\foreach \label in {1,2,3} {
				\draw[line width=2, color=blue] (center1) -- (leaf\label);
			}
			\foreach \label in {4,5,6} {
				\draw[line width=2, color=blue] (center2) -- (leaf\label);
			}
			
			% Dashed circle through leaves
		\end{tikzpicture}
		\caption{Here, $T$ is a nice tree with $6$ leaves and two degree-$4$ intermediate nodes. The above depicts the tree tensor network $\TN_{T}$ with all leaves ordered on a circle.}
		\label{fig:tree-tensor-network}
	\end{center}
\end{figure}
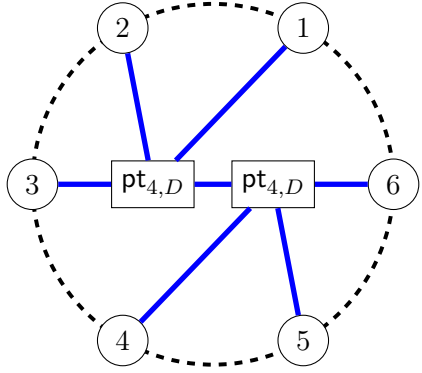

In particular, we identify the leaves of $T$ as the boundary nodes, and consider a planar drawing of $T$ such that all intermediate nodes are inside a circle, and all leaves are on the boundary circle. We then number all the leaves on the boundary circle following their ordering on the cycle (we start with an arbitrary leaf and number it as the first leaf, and continue counter-clockwise through the cycle). We also order the indices in $\State(T)$ so that the $i$-th qudit (or leg, if we interpret $\State(T)$ as a tensor) corresponds to the $i$-th leaf of $T$.

Throughout this section, we will always use $T$ to denote a nice tree with $n$ leaves and edge weight $\ln D$ for some prime power $D$ such that $D \ge n$.

Next, we show how to modify $\TN_{T}$ to obtain our construction of $\caD_{T}$ and $\caD_{T,e}$. Let $T$ be a nice tree and $e$ be an edge in $T$, $f \colon [D] \to \{0,1\}$ and $S \subseteq [D]$.

\paragraph*{Modifying a single edge of $\TN_T$.} We let $\TN_{T,e;f,S}$ be the tensor network obtained by putting the tensor $\sum_{i \in S} (-1)^{f(i)} \cdot \spz{i}\otimes \spz{i}$ on the edge $e$ in $\TN_{T}$, and $\State(T,e;f,S)$ be the corresponding normalized quantum state. Note that when $f$ is the constant $\mathbf{0}$ function and $S = [D]$, $\TN_{T,e;f,S}$  and $\State(T,e;f,S)$ are just $\TN_{T}$ and $\State(T)$, respectively. 

We are now finally ready to define $\caD_{T}$ and $\caD_{T,e}$.

\paragraph*{Defining ensembles of quantum states $\caD_{T}$ and $\caD_{T,e}$.} Assuming the existence of quantum-secure one-way functions, we let $\bF\colon [D] \to \{0,1\}$ be a quantum-secure pseudorandom function and $\bP \colon [D] \to [D]$ be a quantum-secure pseudorandom permutation. We define the following two ensembles $\caD_{T}$ and $\caD_{T,e}$:
\begin{enumerate}
	\item (The distribution $\caD_{T}$) Draw $f \sim \bF$, output
	\[
	\State(T,e; f,[D]).
	\]
	
	\item (The distribution $\caD_{T,e}$) Draw $f \sim \bF$, $p \sim \bP$, set $S = \{ p(i) : i \in \left[\sqrt{D}\right] \}$, output
	\[
	\State(T,e;f,S).
	\]
\end{enumerate}

To show~\cref{theo:PSE-tree}, we will first show in~\cref{sec:pe_link-tree-RT} that $\caD_{T}$ and $\caD_{T,e}$ has entropy structure $T$ and $T_{[e]}$, respectively; then in~\cref{sec:pe_link-tree-ind}, we will show $\caD_{T}$ and $\caD_{T,e}$ are computationally indistinguishable.

\subsubsection{Proof of RT formula entanglement scaling}\label{sec:pe_link-tree-RT}

To show that $\caD_{T}$ and $\caD_{T,e}$ has entropy structure $T$ and $T_{[e]}$, respectively, it suffices to prove the following lemma.

\begin{lemma}\label{lm:RT-general}
	$\State(T,e;f,S)$ has holographic entropy structure $T_{e,S}$, where $T_{e,S}$ is a tree that is identical to $T$ except for edge $e$ having weight $\ln |S|$ instead of $\ln D$.
\end{lemma}

To prove \cref{lm:RT-general}, we need the following decomposition procedure. Given a nice tree $T$ with $n$ leaves and an edge $e$ from $T$, let $A \subseteq [n]$ and $A^c = [n] \setminus A$.

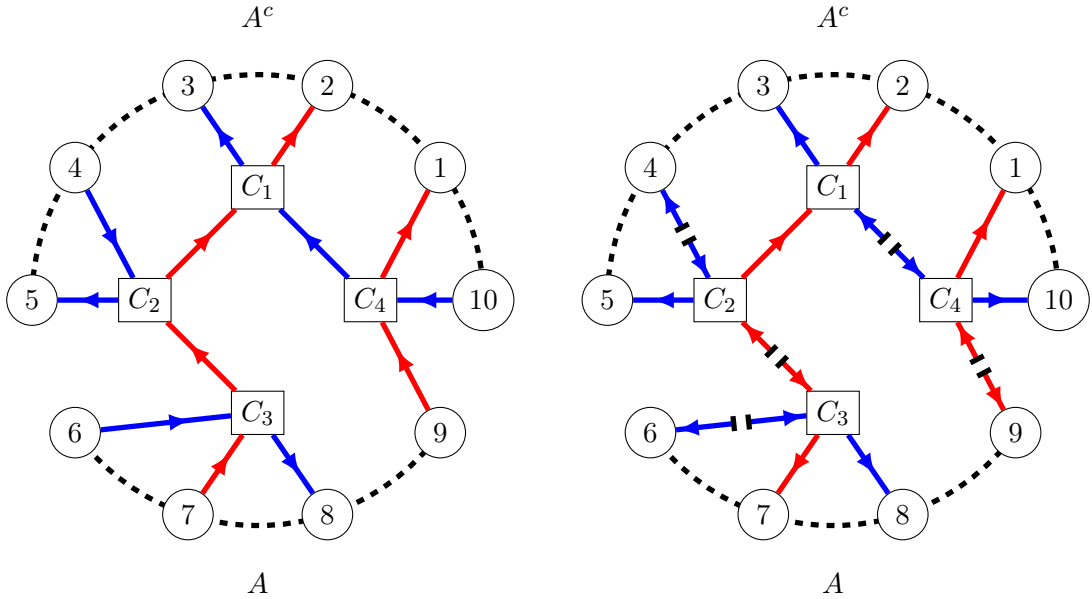
\begin{figure}[h]
	\begin{center}
		\begin{tikzpicture}[scale=1.0, line1/.style={decoration={
					markings,
					mark=at position 0.7 with {\arrow{Latex[length=3mm]}}},postaction={decorate}},
			line2/.style={decoration={
					markings,
					mark=at position 0.3 with {\arrowreversed{Latex[length=3mm]}}},postaction={decorate}},
			line3/.style={decoration={
					markings,
					mark=at position .0 with {\arrowreversed{Latex[length=3mm]}},
					mark=at position .5 with {\fill[white] (-0.1,-0.1) rectangle (0.1,0.1); \draw[black] (-0.1,-0.1) -- (-0.1,0.1); \draw[black] (0.1,-0.1) -- (0.1,0.1);},
					mark=at position 1.0 with {\arrow{Latex[length=3mm]}},
				},postaction={decorate}},
			]
			% Nodes
			
			\node[rectangle,draw,fill=white] (center1) at (0,1.5)  {$C_1$};
			\node[rectangle,draw,fill=white] (center2) at (-1.5,0) {$C_2$};
			\node[rectangle,draw,fill=white] (center3) at (0,-1.5) {$C_3$};
			\node[rectangle,draw,fill=white] (center4) at (1.5,0)  {$C_4$};
			
			% Leaves
			\foreach \angle/\label in {36/1, 72/2, 108/3, 144/4, 180/5, 216/6, 252/7, 288/8, 324/9, 360/10} {
				\node[circle,draw,fill=white] (leaf\label) at (\angle:3cm) {\label};
			}
			
			% Connections
			\draw[line width=2, red, line1] (center2) --  (center1);
			\draw[line width=2, red, line2] (center2) -- (center3);
			\draw[line width=2, blue, line2] (center1) -- (center4);
			
			\draw[line width=2, red, line1] (center1) -- (leaf2);
			\draw[line width=2, blue, line1] (center1) -- (leaf3);
			\draw[line width=2, red, line2](center3) -- (leaf7);
			
			\draw[line width=2, blue, line2] (center2) -- (leaf4);
			\draw[line width=2, blue, line1] (center2) -- (leaf5);
			\draw[line width=2, blue, line2] (center3) -- (leaf6);
			\draw[line width=2, blue, line1] (center3) -- (leaf8);
			\draw[line width=2, blue, line2] (center4) -- (leaf10);

			\draw[line width=2, red, line2] (center4) -- (leaf9);
			\draw[line width=2, red, line1] (center4) -- (leaf1);
			
			\begin{scope}[on background layer]
				\draw[dashed, line width=2] (leaf10) arc (0:144:3cm);
				\draw[dashed, line width=2] (leaf4) arc (144:180:3cm);
				\draw[dashed, line width=2] (leaf6) arc (216:288:3cm);
				\draw[dashed, line width=2] (leaf8) arc (288:324:3cm);
			\end{scope}

			\node[below] at (0,-3.5) {$A$};
			\node[above] at (0,3.5) {$A^c$};
			
			% Dashed circle through leaves
		\end{tikzpicture}
		\qquad
		\begin{tikzpicture}[scale=1.0, line1/.style={decoration={
					markings,
					mark=at position 0.7 with {\arrow{Latex[length=3mm]}}},postaction={decorate}},
			line2/.style={decoration={
					markings,
					mark=at position 0.7 with {\arrowreversed{Latex[length=3mm]}}},postaction={decorate}},
			line3/.style={decoration={
					markings,
					mark=at position .0 with {\arrowreversed{Latex[length=3mm]}},
					mark=at position .5 with {\fill[white] (-0.1,-0.1) rectangle (0.1,0.1); \draw[black] (-0.1,-0.1) -- (-0.1,0.1); \draw[black] (0.1,-0.1) -- (0.1,0.1);},
					mark=at position 1.0 with {\arrow{Latex[length=3mm]}},
				},postaction={decorate}},
			]
			% Nodes
			\node[rectangle,draw,fill=white] (center1) at (0,1.5)  {$C_1$};
			\node[rectangle,draw,fill=white] (center2) at (-1.5,0) {$C_2$};
			\node[rectangle,draw,fill=white] (center3) at (0,-1.5) {$C_3$};
			\node[rectangle,draw,fill=white] (center4) at (1.5,0)  {$C_4$};
			
			% Leaves
			\foreach \angle/\label in {36/1, 72/2, 108/3, 144/4, 180/5, 216/6, 252/7, 288/8, 324/9, 360/10} {
				\node[circle,draw,fill=white] (leaf\label) at (\angle:3cm) {\label};
			}
			
			% Connections
			\draw[line width=2, red, line1] (center2) --  (center1);
			\draw[line width=2, red, line3] (center2) -- (center3);
			\draw[line width=2, blue, line3] (center1) -- (center4);
			
			\draw[line width=2, red, line1] (center1) -- (leaf2);
			\draw[line width=2, blue, line1] (center1) -- (leaf3);
			\draw[line width=2, red, line1](center3) -- (leaf7);
			
			\draw[line width=2, blue, line3] (center2) -- (leaf4);
			\draw[line width=2, blue, line1] (center2) -- (leaf5);
			\draw[line width=2, blue, line3] (center3) -- (leaf6);
			\draw[line width=2, blue, line1] (center3) -- (leaf8);
			\draw[line width=2, blue, line1] (center4) -- (leaf10);

			\draw[line width=2, red, line3] (center4) -- (leaf9);
			\draw[line width=2, red, line1] (center4) -- (leaf1);
			
			\begin{scope}[on background layer]
				\draw[dashed, line width=2] (leaf10) arc (0:144:3cm);
				\draw[dashed, line width=2] (leaf4) arc (144:180:3cm);
				\draw[dashed, line width=2] (leaf6) arc (216:288:3cm);
				\draw[dashed, line width=2] (leaf8) arc (288:324:3cm);
			\end{scope}

			\node[below] at (0,-3.5) {$A$};
			\node[above] at (0,3.5) {$A^c$};
			
			% Dashed circle through leaves
		\end{tikzpicture}
		\caption{Here, $T$ is a nice tree with $10$ leaves and $4$ degree-$4$ intermediate nodes. The above depicts the tree tensor network $\TN_{T}$ with all leaves ordered on a circle. The boundary is partitioned into two parts $A = \{6,7,8,9\}$ and $A^c = \{1,2,3,4,5,10\}$. The min-cut between $A$ and $A^c$ is $2$, so by the max-flow min-cut theorem, we can find $2$ edge-disjoint paths going from $A$ to $A^c$, which are colored red in the graph on the left. In a path covering, we also cover the remaining vertices and edges by $3$ other edge-disjoint paths, which are colored blue in the graph on the left.
		In the graph on the right, we cut each of the 5 paths in the middle, reorient the directions of all edges to be from the cutting points to both end points.}

		\label{fig:flow-on-tree}
	\end{center}
\end{figure}

\paragraph{Decomposition of a nice tree into a path covering.} Recall that all edge weights from $T$ are $\ln D$, which implies that the minimum cut between $A$ and $A^c$ on $T$ is an integer multiple of $\ln D$. Suppose it is $k \cdot \ln D$. Then by the max-flow min-cut theorem, we can find $k$ edge-disjoint paths connecting leaves from $A$ to $A^c$. Let them be $P_1,P_2,\dotsc,P_k$. We then remove all these paths from $T$. We claim that the remaining edges of $T$ can be decomposed into $n/2 - k$ many edge-disjoint paths that either connect leaves from $A$ back to $A$ or leaves from $A^c$ back to $A^c$. 

To see this, we pick an arbitrary remaining edge $e$ from $T$. Note that since all intermediate nodes have even degrees and removing some paths does not change this condition, we can extend $e$ into a path connecting a leaf to another leaf. Since $P_1,\dotsc,P_k$ is already a max flow from $A$ to $A^c$, the path we just constructed cannot be between $A$ and $A^c$ (otherwise, the max flow between $A$ and $A^c$ would be greater than $k \cdot \ln D$). Therefore, the path we just constructed must connect leaves within $A$ or $A^c$. We remove this path from $T$ and continue (note that all intermediate nodes still have even degrees). Since each path connects two distinct leaves of $T$ and there are $n$ leaves in $T$, the process must stop when we construct exactly $n/2 - k$ paths, and these paths cover all edges from $T$ (otherwise, the process would continue, but there are no more leaves to connect between,  contradiction). We call a collection of edge-disjoint leaf-to-leaf paths $P_1,\dotsc,P_{n/2}$ that covers all edges from $T$ a \emph{path covering} of $T$. 

\paragraph*{A unitary mapping from cuts to the boundary state.} Recall that $k \cdot \ln D$ is the minimum cut of $T$ between $A$ and $A^c$. Since $P_1,\dotsc,P_k$ corresponds to a max flow, we can pick edges $e_i \in P_i$ such that $e_1,\dotsc,e_k$ form a minimum cut between $A$ and $A^c$. We then pick some arbitrary edges $e_{k+1},\dotsc,e_{n/2}$ such that $e_{i}$ is on $P_i$ for every $i \in \{k+1,\dotsc,n/2\}$.

We claim that $P_1,\dotsc,P_{n/2}$ together with $e_1,\dotsc,e_{n/2}$ create a unitary mapping from $n$ qudits on the edges $e_1,\dotsc,e_{n/2}$ to the $n$ qudits on the boundary.

In more detail, for every $i$, we (1) cut $P_i$ at $e_i$ and call the two endpoints $\ell_i$ and $r_i$ such that after removing the minimum cut $\{e_j\}_{j \in [k]}$, $\ell_i$ is on the side of $A$ and $r_i$ is on the side of $A^c$ for every $i \in [k]$, and (2) create two ``flows'' from $\ell_i$ and $r_i$ to the two ends of $P_i$, respectively. Note that all intermediate nodes have the same in-flows and out-flows. Hence, an intermediate node can be viewed as a unitary mapping from the incoming edges to the outgoing edges (recall that all intermediate nodes are perfect tensors). The cutting points $\ell_i$ and $r_i$ are the sources of the flow, and the leaves are the sinks of the flow. Therefore, the tensor network after the cut can be interpreted as a unitary mapping from $\ell_i$ and $r_i$ to the leaves. See~\cref{fig:flow-on-tree} for an illustration.

Now, we are ready to prove~\cref{lm:RT-general}, which is restated below in more detail for convenience.

\begin{reminder}{\cref{lm:RT-general}}
	Let $e$ be an edge in $T$, $f \colon [D] \to \{0,1\}$ and $S \subseteq [D]$. For every $A \subseteq [n]$, it holds that
	\[
	S_{A}\left(\State(T,e;f,S)\right) = \mincut_{A,[n] \setminus A}(T_{e,S}),
	\]
	where $ \mincut_{A,[n] \setminus A}(T_{e,S})$ denotes the minimum cut between $A$ and $[n] \setminus A$ in the (weighted) graph $T_{e,S}$.
\end{reminder}

\begin{proofof}{\cref{lm:RT-general}}
	Let $e$ be an edge and $A \subseteq [n]$. In the following, we consider two cases; we will construct the paths $P_1,\dotsc,P_{n/2}$ and edges $e_1,\dotsc,e_{n/2}$ depending on which case we are in.
	
	\paragraph*{Case I: $e$ belongs to some min-cut.} The first case is that there exists a min-cut between $A$ and $A^c$ on $T$ such that $e$ is part of it. Let $e_1,e_2,\dotsc,e_k$ be a min-cut such that $e_1 = e$. In this case, we can find $k$ paths $P_1,\dotsc,P_k$ between $A$ and $A^c$ such that $e_i$ is an edge on $P_i$ for every $i \in [k]$ (Indeed, any max-flow from $A$ to $A^c$ must be saturated on $e_1,\dotsc,e_k$, meaning that the corresponding paths contain $e_1,\dotsc,e_k$, and each $e_i$ is on exactly one path. Note that if there is a path crossing more than one $e_i$'s, then the total flow would be less than $k$). We also find paths $P_{k+1},\dotsc,P_{n/2}$ that connect leaves within $A$ or $A^c$ and select edges $e_i$ from $P_i$ for every $i \in \{k+1,\dotsc,n/2\}$.

	\paragraph*{Case II: $e$ does not belong to any min-cut.} The second case is that $e$ does not belong to any min-cut between $A$ and $A^c$ on $T$. In particular, this means if we remove $e$ from $T$, the min-cut between $A$ and $A^c$ is still $k \cdot \ln D$. This, in turn, implies that we can find $k$ paths $P_1,\dotsc,P_k$ and min-cut $e_1,\dotsc,e_k$ between $A$ and $A^c$ such that $e_i$ is on $P_i$ for every $i \in [k]$, and $e$ is not contained in any of the paths. We then find paths $P_{k+1},\dotsc,P_{n/2}$ that connect leaves within $A$ or $A^c$. Without loss of generality, we assume that $e$ is on path $P_{k+1}$ and let $e_{k+1} = e$. We also select edges $e_{i}$ from $P_i$ for every $i \in \{k+2,\dotsc,n/2\}$.
	
	\newcommand{\istar}{i^{\star}}
	
	\paragraph*{Calculating the entropy.} Now we are ready to calculate the entropy $S_A(\State(T,e;f,S))$. For every $i \in [n/2]$, we set $S_i = S$ if $e_i = e$ and $S_i = [D]$ otherwise. Note that in the first case above, $S_1 = S$, and in the second case above, $S_{k+1} = S$. For notational convenience, we set $\istar = 1$ in the first case, and $\istar = k+1$ in the second case.
	
	Let $\ell_i$ and $r_i$ be the two cut ends of $e_i$ for every $i \in [n/2]$, and the max flow from $A$ to $A^c$ flow from $\ell_i$ to $r_i$ for every $i \in [k]$. Let $W_A$ and $W_{A^c}$ be the set of $i \in [n/2] \setminus [k]$ such that $P_i$ connects within $A$ or $A^c$, respectively.
	
	We have
	\[
	\State(T,e;f,S) = (U_A \otimes U_{A^c}) \left( \bigotimes_{i \in [n/2] \setminus \{\istar\} } \frac{1}{\sqrt{|S_i|}} \sum_{j \in S_i} \spz{j}_{\ell_i}\spz{j}_{r_i} \otimes \frac{1}{\sqrt{|S_{\istar}|}} \sum_{j \in S_{\istar}} (-1)^{f(j)} \spz{j}_{\ell_{\istar}}\spz{j}_{r_{\istar}} \right).
	\]
	
	Where $U_A$ is the unitary mapping from $\ell_1,\dotsc,\ell_k$ and $\ell_i,r_i$ for every $i \in W_A$ to indices in $A$, and $U_{A^c}$ is the unitary mapping from $r_1,\dotsc,r_k$ and $\ell_i,r_i$ for every $i \in W_{A^c}$ to indices in $A^c$.
	
	In particular, the above means the entropy of $\State(T,e;f,S)$ across $A$ and $A^c$ equals the entropy of 
	\[
	\bigotimes_{i \in [n/2] \setminus \{\istar\} } \frac{1}{\sqrt{|S_i|}} \sum_{j \in S_i} \spz{j}_{\ell_i}\spz{j}_{r_i} \otimes \frac{1}{\sqrt{|S_{\istar}|}} \sum_{j \in S_{\istar}} (-1)^{f(j)} \spz{j}_{\ell_{\istar}}\spz{j}_{r_{\istar}}
	\]
	across $\ell_1,\dotsc,\ell_k$ and $\ell_i,r_i$ for every $i \in W_A$ and $r_1,\dotsc,r_k$ and $\ell_i,r_i$ for every $i \in W_{A^c}$. This entropy can be directly calculated as $\sum_{i \in [k]} \ln |S_i|$, which is exactly the weight of sums of $e_1,\dotsc,e_k$, which is in turn the min-cut between $A$ and $A^c$ in $T$.
\end{proofof}

\subsubsection{Computational Indistinguishability}\label{sec:pe_link-tree-ind}

Next, we show that $\caD_{T}$ and $\caD_{T,e}$ are computationally indistinguishable.

\begin{lemma}
	$\caD_{T}$ and $\caD_{T,e}$ are computationally indistinguishable.
\end{lemma}

\begin{proof}
	Let $A$ and $A^c$ be the partition of boundary vertices $[n]$ when removing the edge $e$ from $T$. Clearly, $e$ is the min-cut between $A$ and $A^c$ in $T$.
	
	Following the proof of~\cref{lm:RT-general} applied to the cut $A$ and $A^c$ and edge $e$ (here, we must be in the first case since $e$ is the min-cut), we have
	\begin{equation}\label{eq:State-decomp}
		\State(T,e;f,S) = (U_A \otimes U_{A^c}) \left( \bigotimes_{i \in [n/2] \setminus \{1\} } \frac{1}{\sqrt{|S_i|}} \sum_{j \in S_i} \spz{j}_{\ell_i}\spz{j}_{r_i} \otimes \frac{1}{\sqrt{|S_{1}|}} \sum_{j \in S_{1}} (-1)^{f(j)} \spz{j}_{\ell_{1}}\spz{j}_{r_{1}} \right),
	\end{equation}
	where all the $S_{i}$, $\ell_i$, $r_i$ are defined as in~\cref{lm:RT-general}.
	
	Now, let $t$ be an arbitrary polynomial in $n$, and $\caD_{\sf full}$ and $\caD_{\sf subset}$ be the two distributions from~\cref{cor:pseudo-ent} with $k = \sqrt{D}$.
Let
\[
\sigma_{\sf large} = \Ex_{\spz{\phi} \sim \caD_{\sf full}} \left[ \optk{\phi}{t} \right]
\]
and
\[
\sigma_{\sf small} = \Ex_{\spz{\phi} \sim \caD_{\sf subset}} \left[ \optk{\phi}{t} \right].
\]	
\cref{cor:pseudo-ent} implies that $\sigma_{\sf large}$ and $\sigma_{\sf small}$ are computationally indistinguishable against polynomial-time quantum adversaries.

By our decomposition of $\State(T,e;f,S)$ from~\cref{eq:State-decomp}, it follows that there exists a fixed (and polynomial-time computable given $T$ and $e$) state $\Psi$ such that
\[
\Ex_{\spz{\phi} \sim \caD_{T}} \left[ \optk{\phi}{t} \right] = (U_A \otimes U_{A^c})^{\otimes t}  \left(\Psi \otimes \sigma_{\sf large} \right) (U_A \otimes U_{A^c})^{\dagger t}
\]
and
\[
\Ex_{\spz{\phi} \sim \caD_{T,e}} \left[ \optk{\phi}{t} \right] = (U_A \otimes U_{A^c})^{\otimes t}  \left(\Psi \otimes \sigma_{\sf small} \right) (U_A \otimes U_{A^c})^{\dagger t}.
\]

Hence, these two mixed states are also computationally indistinguishable against polynomial-time quantum adversaries.
\end{proof}

\subsubsection{Extension to public-key pseudoentangled holographic states} \label{sec:pe_link-tree-public-key}

Now we briefly discuss how to extend our construction to the public-key version of pseudoentangled holographic states~\cite{bouland2023public}. We first define public-key pseudoentangled holographic states, following~\cite{bouland2023public}.

\newcommand{\sfKPhi}{\sfK^{\Phi}}
\newcommand{\sfKPsi}{\sfK^{\Psi}}
\newcommand{\sfGen}{\mathsf{Gen}}

\begin{definition}[Public-key pseudoentangled holographic states with entropy gap]\label{defi:pk-pe-states}
	Let $\lambda$ be the security parameter. Let $\H = \{ \H_{\lambda} \}_{\lambda \in \posN}$, $\sfKPhi = \{ \sfKPhi_{\lambda} \}_{\lambda \in \posN}$ and $\sfKPsi = \{ \sfKPsi_{\lambda} \}_{\lambda \in \posN}$ be a family of Hilbert spaces and two families of key spaces. Let $\{G_k\}_{k \in \sfKPhi \cup \sfKPsi}$ be a family of keyed weighted graphs indexed by $\lambda$. Two families of quantum states $\{\spz{\Phi}_{k} \in \rmS(\H)\}_{k \in \sfKPhi}$ and $\{\spz{\Psi}_{k} \in \rmS(\H)\}_{k \in \sfKPsi}$ (parameterized by $\lambda$) form a public-key pseudoentangled holographic state ensemble (PES) with exact (resp.~approximate) entropy structure and gap $A$ vs $B$ w.r.t. to cut $S$, if the following three conditions hold:
	
	\begin{enumerate}
		\item There is a polynomial-time quantum algorithm $\sfGen$ that given key $k \in \sfKPhi \cup \sfKPsi$, outputs a quantum state $\spz{\psi_k} \in \rmS(\H)$ that has entropy structure $G_k$. Suppose that $\spz{\psi_k}$ has $n$ $D$-dimensional qudits.
		
		\item For any polynomial-time quantum algorithm $A$, it holds that
		\[
		\left| \Pr_{k \leftarrow \sfKPsi_{\lambda}}[ A(k) = 1] - \Pr_{k \leftarrow \sfKPhi_{\lambda}}[ A(k) = 1] \right| \le \negl(\lambda)\mathdot
		\]
		
		\item The following statements are true:
		\begin{itemize}
			\item $\Pr_{k \leftarrow \sfKPsi_{\lambda}}[ \mincut_{S, [n] \setminus S}(G_k) \le A ] \ge 1 - \negl(\lambda)$.
			\item $\Pr_{k \leftarrow \sfKPhi_{\lambda}}[ \mincut_{S, [n] \setminus S}(G_k) \ge B ] \ge 1 - \negl(\lambda)$.
		\end{itemize}
	\end{enumerate}
	
\end{definition}

We will prove the following.
\begin{theorem}[Public-key pseudoentangled holographic tree states]\label{theo:PSE-pk-tree}
	Let $\eps \in (0,1)$. Let $T$ be a nice tree with $n$ leaves and edge weight $\ln D \ge n^{\Omega(1)}$. Assuming the standard LWE assumption holds,\footnote{see~\cite[Assumption~2.18]{bouland2023public} for a formal definition} we have:
	\begin{itemize}
		\item There are two ensembles of quantum states $\caD_{\sf low}$ and $\caD_{\sf high}$ that constitute a pseudoentangled holographic state ensemble with exact entropy structure and gap $n^{\eps}$ vs. $\Omega(n)$ w.r.t. to some cut $S$.
	\end{itemize}
\end{theorem}

To modify our previous (private-key) construction to be public-key, we will make use of the following public-key pseudo-entangled states by~\cite{bouland2023public}.

\begin{lemma}\label{lm:bouland-pk}
	Assuming the standard LWE assumption holds, for any $\eps \in (0,1)$, there are two families of key spaces $\sfKPhi = \{ \sfKPhi_{\lambda} \}_{\lambda \in \posN}$ and $\sfKPsi = \{ \sfKPsi_{\lambda} \}_{\lambda \in \posN}$ such that the following holds:
	
	\begin{itemize}
		\item There is a polynomial-time quantum algorithm $\sfGen$ that given a key $k \in \sfKPhi_\lambda \cup \sfKPsi_\lambda$, outputs a $2n$-qubit quantum state $\spz{\psi_k}$.
		
		\item For any polynomial-time quantum algorithm $A$, it holds that
		\[
		\left| \Pr_{k \leftarrow \sfKPsi_{\lambda}}[ A(k) = 1] - \Pr_{k \leftarrow \sfKPhi_{\lambda}}[ A(k)] \right| \le \negl(\lambda)\mathdot
		\]
		
		\item The followings are true:
		\begin{itemize}
			\item $\Pr_{k \leftarrow \sfKPsi_{\lambda}}[ S(\psi_k)_{[\lambda]} \le \lambda^\eps ] \ge 1 - \negl(\lambda)$.
			\item $\Pr_{k \leftarrow \sfKPhi_{\lambda}}[ S(\psi_k)_{[\lambda]} \ge \Omega(\lambda) ] \ge 1 - \negl(\lambda)$.
		\end{itemize}
	
		Here, $S(\psi_k)_{[\lambda]}$ denotes the entropy of $\psi_k$ across the bipartition $[\lambda]$ and $[2\lambda] \setminus [\lambda]$.
		
	\end{itemize}
\end{lemma}

Now, we are ready to specify the construction of our public-key pseudoentangled holographic states. Let $\gamma \in (0,1)$ be a constant. We let $D = 2^{\ceil{n^\gamma}}$.

Recall that we used $\TN_{T,e;f,S}$ to denote the tensor network obtained by putting the tensor $\sum_{i \in S} (-1)^{f(i)} \cdot \spz{i}\otimes \spz{i}$ on the edge $e$ in $\TN_{T}$, and $\State(T,e;f,S)$ to denote the corresponding normalized quantum state. Similarly, we use $\TN_{T,e;\spz{\psi}}$ to denote the tensor network obtained by putting the $2 \log D$-qubit state $\spz{\psi}$ (interpreted as a $2$-leg tensor of local dimension $D$) on the edge $e$ in $\TN_{T}$, and $\State(T,e; \spz{\psi} )$ to denote the corresponding normalized quantum state.

Letting $\lambda = n^\gamma$. Now, we are ready to define our ensembles of public-key pseudoentangled holographic states:

\begin{enumerate}
	\item[$(\caD_{\sf low})$] Draw $k \leftarrow \sfKPsi_{\lambda}$, run $\mathsf{Gen}(k)$ to obtain $\spz{\psi_k}$, and output $\State(T,e; \spz{\psi_k} )$.
	
	\item[$(\caD_{\sf high})$] Draw $k \leftarrow \sfKPhi_{\lambda}$, run $\mathsf{Gen}(k)$ to obtain $\spz{\psi_k}$, and output $\State(T,e; \spz{\psi_k} )$.
\end{enumerate}

From~\cref{lm:RT-general}, the states from $\caD_{\sf low}$ and $\caD_{\sf high}$ has entropy structure specified by a corresponding tree, and the algorithm $\mathsf{Gen}$ from Condition~(i) of~\cref{defi:pk-pe-states} can be constructed similar following the proof of~\cref{lm:RT-general}. Condition~(ii) and~(iii) of~\cref{defi:pk-pe-states} follow all straightforwardly from~\cref{lm:bouland-pk}, which proves~\cref{theo:PSE-pk-tree}.

\subsection{Pseudoentangled holographic states via Random Stabilizer Tensor Networks} \label{sec:pe_link-planar}

Next, we describe another construction of pseudoentangled holographic states via random stabilizer tensor networks. This construction works on any planar graph and can also be made public-key as the previous one. Formally, for a ``nice'' family of bulk geometry graph (which will be defined formally in~\cref{sec:pe_link-planar-cond}), we prove:

\begin{theorem}[Pseudoentangled holographic states]\label{theo:PSE-general-apx}
	Let $G = \{G_n\}_{n \in \N}$ be a family of nice bulk geometry graphs, where $G_n$ has $n$ boundary nodes and has edge weight $\ln q_n \ge \omega(\ln n)$ for some prime power $q_n$. Let $A \subseteq [n]$ be a continuous segment on the boundary of $G_n$ and let $\gamma = \{e_1,\dotsc,e_t\}$ be a minimum cut between $A$ and $A^c = [n] \setminus A$ on $G$.
	
	Let $G_{n,\gamma}$ be the graph that is the same as $G_n$ except for the weights of edges $e_1,\dotsc,e_t$ are reduced to $\ln q / 2$ from $\ln q$. Assuming the existence of quantum-secure one-way functions, there are two ensembles of quantum states $\caG_{G_n,\gamma}$ and $\caH_{G_n,\gamma}$ that constitute a pseudoentangled holographic state ensemble with exact entropy structure $G_n$ vs. $G_{n,\gamma}$.
\end{theorem}

In the rest of the section, in~\cref{sec:pe_link-planar-cond}, we formally define a family of nice bulk geometry graphs. Then, in~\cref{sec:pe_link-planar-states}, we formally define the ensemble of quantum states $\caG_{G_n,\gamma}$ and $\caH_{G_n,\gamma}$ and prove that they are computationally indistinguishable. In~\cref{sec:pe_link-planar-rt}, we show these two ensembles satisfy the RT formula approximately. Finally, in~\cref{sec:pe_link-planar-public-key}, we discuss its generalization to the public-key version.

\subsubsection{Conditions on the graph families}\label{sec:pe_link-planar-cond}

\paragraph*{Bulk Geometry Graphs.} We say a weighted planar graph $G$ is a \emph{bulk geometry graph} (here, we assume $G$ also comes with a known planar drawing) if the following holds:

\begin{enumerate}
	\item Every vertex either has even degree or degree exactly one and all edges have the same weight. We call the degree-one nodes the \emph{boundary} nodes, and other nodes the \emph{bulk} nodes. Suppose there are $n$ boundary nodes and $n_{\bulk}$ bulk nodes.
	
	\item The graph $G$ is planar, and there is a planar embedding in which one can draw a circle connecting all the boundary nodes such that all bulk nodes are strictly inside the boundary circle. We then number all the boundary vertices on the boundary circle following their ordering on the cycle (we start with an arbitrary boundary vertex and number it as the boundary vertex $1$, and continue counter-clockwise through the cycle).
\end{enumerate}

We also consider the dual graph $G^*$ of $G$, whose vertices are the faces of $G$ inside the boundary circle (i.e., the surrounding area outside of the boundary circle is not a vertex in $G^*$). Two vertices of $G^*$ are connected if their corresponding faces in $G$ share a common edge; see~\autoref{fig:bulk-geometry-graph} for an example of a bulk geometry graph $G$ and its dual $G^*$.

\begin{figure}
	\begin{center}
		\resizebox{0.45\linewidth}{!}{%
			\begin{tikzpicture}[scale=1.0, line1/.style={},
				line2/.style={},
				line3/.style={},
				]
				% Nodes
				\foreach \i in {1,12,11,5,6,7}{
					\pgfmathtruncatemacro{\tmp}{\i*30}
					\node[circle,draw,fill=white,line width=1.5pt] (leaf\i) at (\tmp:3.3cm) {\i};
				}
				
				\foreach \i in {2,3,4,8,9,10}{
					\pgfmathtruncatemacro{\tmp}{\i*30}
					\node[circle,draw,fill=white,line width=1.5pt] (leaf\i) at (\tmp:3.3cm) {\i};
				}
				
				\foreach \x in {1,...,3}{
					\foreach \y in {1,...,3}{
						\pgfmathtruncatemacro{\tmp}{(\x-1)*3+\y}
						\node[rectangle,draw,fill=white, line width=1.5pt] (center-\x-\y) at (1.5*\y-3, 3-1.5*\x)  {$C_{\tmp}$};
					}
				}
				\newcommand{\drawth}[4]{
					% Calculate the centroid
					\path (barycentric cs:#1=1,#2=1,#3=1) coordinate (centroid);
					% Draw the centroid
					\draw[line width=2] (#1) -- (#2);
					\draw[line width=2] (#1) -- (#3);
					\node (face-#4) at (centroid) {$f_{#4}$};
				}			
				
				\newcommand{\drawf}[5]{
					% Calculate the centroid
					\path (barycentric cs:#1=1,#2=1,#3=1,#4=1) coordinate (centroid);
					% Draw the centroid
					\draw[line width=2] (#1) -- (#3);
					\draw[line width=2] (#2) -- (#4);
					\node (face-#5) at (centroid) {$f_{#5}$};
				}
				
				\foreach \x in {1,...,3}{
					\foreach \y in {1,...,3}{
						\pgfmathtruncatemacro{\a}{\x+1}
						\pgfmathtruncatemacro{\b}{\y+1}
						\pgfmathtruncatemacro{\tmp}{13+(\x-1)*2+(\y-1)}
						\ifnum \x < 3
						\draw[line width=2] (center-\x-\y) -- (center-\a-\y);
						\ifnum \y < 3
						\drawf{center-\x-\y}{center-\a-\y}{center-\x-\b}{center-\a-\b}{\tmp}
						\fi
						\fi
						\ifnum \y < 3
						\draw[line width=2] (center-\x-\y) -- (center-\x-\b);
						\fi
					}
					12			}
				
				\drawth{center-1-1}{leaf4}{leaf5}{1};
				\drawf{center-1-1}{center-1-2}{leaf4}{leaf3}{2};
				\drawf{center-1-2}{center-1-3}{leaf3}{leaf2}{3};
				\drawth{center-1-3}{leaf2}{leaf1}{4};
				\drawf{center-1-3}{center-2-3}{leaf1}{leaf12}{5};
				\drawf{center-2-3}{center-3-3}{leaf12}{leaf11}{6};
				\drawth{center-3-3}{leaf11}{leaf10}{7};
				\drawf{center-3-3}{center-3-2}{leaf10}{leaf9}{8};
				\drawf{center-3-2}{center-3-1}{leaf9}{leaf8}{9};
				\drawth{center-3-1}{leaf7}{leaf8}{10};
				\drawf{center-3-1}{center-2-1}{leaf7}{leaf6}{11};
				\drawf{center-2-1}{center-1-1}{leaf6}{leaf5}{12};
				
				\begin{scope}[on background layer]
					\draw[dashed, line width=2] (0,0) circle [radius=3.3cm];
				\end{scope}
				
				% Dashed circle through leaves
			\end{tikzpicture}
		}
		\qquad
		\resizebox{0.45\linewidth}{!}{%		
			\begin{tikzpicture}[scale=1.0, dot/.style={circle,fill=#1,inner sep=3pt,outer sep=0pt}
				]
				% Nodes
				\foreach \i in {1,12,11,5,6,7}{
					\pgfmathtruncatemacro{\tmp}{\i*30}
					\node[inner sep=0pt, outer sep=0pt] (leaf\i) at (\tmp:3.3cm) {};
				}
				
				\foreach \i in {2,3,4,8,9,10}{
					\pgfmathtruncatemacro{\tmp}{\i*30}
					\node[inner sep=0pt, outer sep=0pt] (leaf\i) at (\tmp:3.3cm) {};
				}
				
				\foreach \x in {1,...,3}{
					\foreach \y in {1,...,3}{
						\pgfmathtruncatemacro{\tmp}{(\x-1)*3+\y}
						\ifnum\tmp=1
						\node[inner sep=0pt, outer sep=0pt] (center-\x-\y) at (1.5*\y-3, 3-1.5*\x)  {};
						\else
						\ifnum\tmp=5
						\node[inner sep=0pt, outer sep=0pt] (center-\x-\y) at (1.5*\y-3, 3-1.5*\x)  {};
						\else
						\ifnum\tmp=6
						\node[inner sep=0pt, outer sep=0pt] (center-\x-\y) at (1.5*\y-3, 3-1.5*\x)  {};
						\else
						\node[inner sep=0pt, outer sep=0pt] (center-\x-\y) at (1.5*\y-3, 3-1.5*\x)  {};
						\fi\fi\fi
						
					}
				}
				\newcommand{\drawth}[4]{
					% Calculate the centroid
					\path (barycentric cs:#1=1,#2=1,#3=1) coordinate (centroid);
					% Draw the centroid
					\draw[line width=2] (#1) -- (#2);
					\draw[line width=2] (#1) -- (#3);
					\node[dot] (face-#4) at (centroid) {};
				}			
				
				\newcommand{\drawf}[5]{
					% Calculate the centroid
					\path (barycentric cs:#1=1,#2=1,#3=1,#4=1) coordinate (centroid);
					% Draw the centroid
					\draw[line width=2] (#1) -- (#3);
					\draw[line width=2] (#2) -- (#4);
					\node[dot] (face-#5) at (centroid) {};
				}
				% Macro to draw the line from the reflection of C across AB to C
				
				\newcommand{\drawReflectionLine}[3]{
					% Calculate the reflection of point #3 across the line defined by #1 and #2
					% This uses the 'calc' library for coordinate calculations
					\coordinate (reflected) at ($2*(#1)!0.5!(#2)-(#3)$);
					
					% Draw the reflection point
					
					% Draw line from the reflected point to C
					\draw [dashed, violet, line width=2] (reflected) -- (#3);
				}
				
				\foreach \x in {1,...,3}{
					\foreach \y in {1,...,3}{
						\pgfmathtruncatemacro{\a}{\x+1}
						\pgfmathtruncatemacro{\b}{\y+1}
						\pgfmathtruncatemacro{\tmp}{13+(\x-1)*2+(\y-1)}
						\ifnum \x < 3
						\draw[line width=2] (center-\x-\y) -- (center-\a-\y);
						\ifnum \y < 3
						\drawf{center-\x-\y}{center-\a-\y}{center-\x-\b}{center-\a-\b}{\tmp}
						\fi
						\fi
						\ifnum \y < 3
						\draw[line width=2] (center-\x-\y) -- (center-\x-\b);
						\fi
					}
				}
				
				\drawth{center-1-1}{leaf4}{leaf5}{1};
				\drawf{center-1-1}{center-1-2}{leaf4}{leaf3}{2};
				\drawf{center-1-2}{center-1-3}{leaf3}{leaf2}{3};
				\drawth{center-1-3}{leaf2}{leaf1}{4};
				\drawf{center-1-3}{center-2-3}{leaf1}{leaf12}{5};
				\drawf{center-2-3}{center-3-3}{leaf12}{leaf11}{6};
				\drawth{center-3-3}{leaf11}{leaf10}{7};
				\drawf{center-3-3}{center-3-2}{leaf10}{leaf9}{8};
				\drawf{center-3-2}{center-3-1}{leaf9}{leaf8}{9};
				\drawth{center-3-1}{leaf7}{leaf8}{10};
				\drawf{center-3-1}{center-2-1}{leaf7}{leaf6}{11};
				\drawf{center-2-1}{center-1-1}{leaf6}{leaf5}{12};
				
				\draw [color=red, line width=2] (face-1) -- (face-2) -- (face-3) -- (face-4) -- (face-5) -- (face-6) -- (face-7) -- (face-8) -- (face-9) -- (face-10) -- (face-11) -- (face-12) -- (face-13) -- (face-14) -- (face-5);
				\draw [color=red, line width=2] (face-1) -- (face-12);
				\draw [color=red, line width=2] (face-2) -- (face-13) -- (face-15) -- (face-9);
				\draw [color=red, line width=2] (face-3) -- (face-14) -- (face-16) -- (face-8);
				\draw [color=red, line width=2] (face-11) -- (face-15) -- (face-16) -- (face-6);
				
				\begin{scope}[on background layer]
					\draw[line width=2] (0,0) circle [radius=3.3cm];
				\end{scope}
				
				% Dashed circle through leaves
			\end{tikzpicture}
		}
		\caption{A bulk geometry graph $G$ with $12$ boundary nodes and $9$ bulk nodes (left) and its dual $G^*$ (right)}
		\label{fig:bulk-geometry-graph}
	\end{center}
\end{figure}
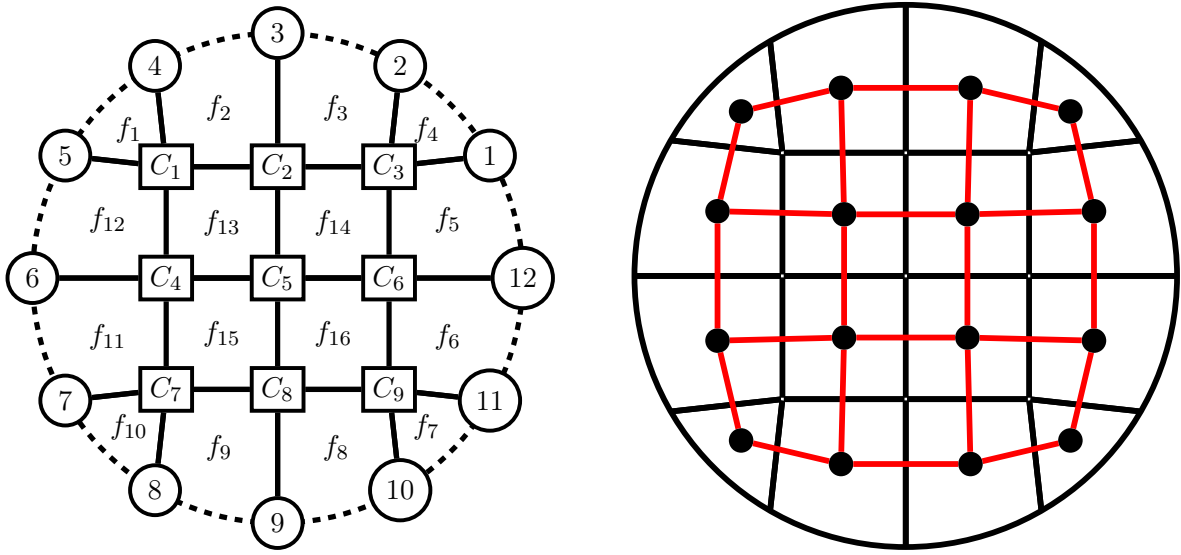

For a weighted graph $G$, we use $\WT{G}$ to denote the unweighted version of $G$ (in which all edges have unit weights). We say a bulk geometry graph $G$ is \emph{nice} if the following conditions hold:

\begin{figure}
	\begin{center}
		\begin{tikzpicture}[scale=1.0, line1/.style={},
			line2/.style={},
			line3/.style={},
			]
			% Nodes
			
			\node[rectangle,draw,fill=white] (center1) at (0,1.5)  {$C_1$};
			\node[rectangle,draw,fill=white] (center2) at (-1.5,0) {$C_2$};
			\node[rectangle,draw,fill=white] (center3) at (0,-1.5) {$C_3$};
			\node[rectangle,draw,fill=white] (center4) at (1.5,0)  {$C_4$};
			
			% Leaves
			
			\foreach \i in {1,...,8}{
				\pgfmathtruncatemacro{\tmp}{\i*45}
				\node[circle,draw,fill=white] (leaf\i) at (\tmp:3cm) {\i};
			}
			
			% Connections
			\draw[line width=2, blue, line1] (center2) --  (center1);
			\draw[line width=2, blue, line3] (center2) -- (center3);
			\draw[line width=2, blue, line3] (center1) -- (center4);
			\draw[line width=2, blue, line3] (center3) -- (center4);
			
			\draw[line width=2, blue, line1] (center1) -- (leaf2);
			\draw[line width=2, blue, line1] (center1) -- (leaf3);
			\draw[line width=2, blue, line1](center3) -- (leaf7);
			
			\draw[line width=2, blue, line3] (center2) -- (leaf4);
			\draw[line width=2, blue, line1] (center2) -- (leaf5);
			\draw[line width=2, blue, line3] (center3) -- (leaf6);	
			
			\draw[line width=2, blue, line3] (center4) -- (leaf8);
			\draw[line width=2, blue, line1] (center4) -- (leaf1);
			
			\begin{scope}[on background layer]
				\draw[dashed, line width=2] (0,0) circle [radius=3cm];
			\end{scope}

			\newcommand{\drawth}[4]{
				% Calculate the centroid
				\path (barycentric cs:#1=1,#2=1,#3=1) coordinate (centroid);
				% Draw the centroid
				\node at (centroid) {#4};
			}			
			
			\newcommand{\drawf}[5]{
				% Calculate the centroid
				\path (barycentric cs:#1=1,#2=1,#3=1,#4=1) coordinate (centroid);
				% Draw the centroid
				\node at (centroid) {#5};
			}
			
			\drawth{center4}{leaf1}{leaf8}{0};
			\drawf{center4}{center3}{leaf7}{leaf8}{1};
			\drawf{center4}{center1}{leaf2}{leaf1}{1};
			\drawf{center1}{center2}{center3}{center4}{2};	
			\drawth{center1}{leaf2}{leaf3}{2};
			\drawf{center1}{center2}{leaf3}{leaf4}{3};	
			\drawth{center2}{leaf4}{leaf5}{4};
			\drawf{center2}{center3}{leaf5}{leaf6}{3};	
			\drawth{center3}{leaf6}{leaf7}{2};
			% Dashed circle through leaves
		\end{tikzpicture}
		\caption{A drawing of the graph $G$ with $8$ boundary nodes and $4$ bulk nodes. We also visualize the distance function on $G^*$ from starting from the region enclosed by center node $C_4$ and boundary nodes $1$ and $8$.}
		\label{fig:conditions-on-graph}
	\end{center}
\end{figure}
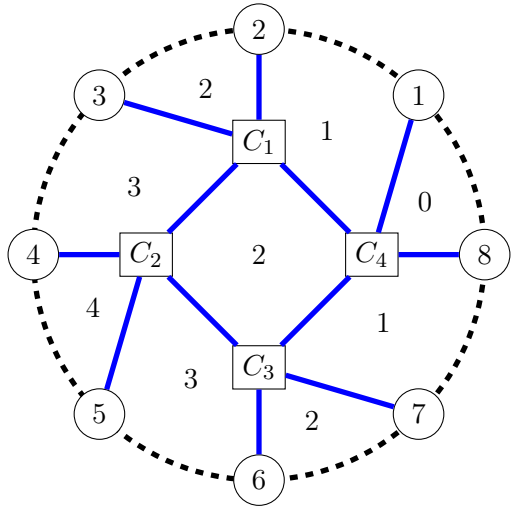

\begin{itemize}
	\item \textbf{Planar.} The graph $G$ is planar, and there is a planar embedding in which one can draw a circle connecting all the boundary nodes in a certain order such that all bulk nodes are strictly inside the boundary circle. 
	
	\item \textbf{The ``negative curvature'' condition from~\cite[Appendix~B]{Pastawski_2015}}: For the dual graph $G^*$ of $G$, for every face $f$ touching the boundary circle of $G$, consider the minimum distance function $d_f$ from $f$ to every other face in $G^*$,.
    Then $d_f$ has no local maximum that is not touching the boundary circle of $G$; see~\cref{fig:conditions-on-graph} for an illustration.
	
\end{itemize}

\paragraph*{Notation.}
Let $G = \{G_n\}_{n \in \N}$ be a family of nice bulk geometry graphs, where $G_n$ has $n$ boundary nodes. We will assume the corresponding planar embedding of $G_n$ is given, and its boundary nodes are labeled from $1$ to $n$ counter-clockwise following the boundary circle. 

Let the edge weight of $G_n$ be $\ln q_n$ for some $q_n \in \N$. We assume that $q_n$ is a prime power and $\ln q_n \ge \omega(\ln n)$. For simplicity, we also assume that the number of bulk vertices in $G_n$ is always at most $n^\beta$, where $\beta > 0$ is an absolute constant.

\subsubsection{Construction of pseudoentangled holographic states}\label{sec:pe_link-planar-states}

We define $\caT_{G_n}$ as the distribution of tensor networks obtained by replacing each bulk node $u$ of $G_n$ by an independent uniformly random stabilizer state with $d_u$ legs and $q_n$ bond dimensions (where $d_u$ is the degree of $u$ in $G_n$). For brevity, below we use $q$ to denote $q_n$.

\newcommand{\caTg}{\caT^{\sf good}}

By~\cref{thm Random stabilizer tensors are perfect} and a union bound over all bulk nodes $u$ in $G_n$, it holds that with probability $1 - n^{-\omega(1)}$ over $T \sim \caT_{G_n}$, all tensors in $T$ are perfect. We call such a tensor $T$ \emph{good} and we let $\caTg_{G_n}$ be the distribution of $T \sim \caT_{G_n}$ conditioning on $T$ being good.

Fix a good $T$, and let $A \subseteq [n]$ be a contiguous segment of the boundary of $G_n$. Let $t \cdot \ln q$ be the minimum cut between $A$ and $A^c = [n] \setminus A$ over $G_n$ (i.e., $t \cdot \ln q = \mincut_{A,[n] \setminus A}(G_n)$).

\paragraph*{A unitary mapping from cuts to the boundary.} Let $\gamma$ be the set of the $t$ edges on a minimum cut between $A$ and $A^c = [n] \setminus A$ over $G_n$. Let these edges be $e_1 = (\ell_1,r_1),\dotsc, e_t = (\ell_t,r_t)$, where the $\ell_i$'s are on the $A$ side and the $r_i$'s are on the $A^c$ side once $\gamma$ is removed from $G_n$.

By~\cite[Appendix~B]{Pastawski_2015}, the first three conditions on $G_n$, and the assumption that all tensors in $T$ are perfect, we can partition $A$ into two parts $A_0,A_1$, $A^c$ into two parts $A^c_0,A^c_1$, and then construct a unitary $P$ from $\gamma$ and $A_0$ to $A_1$, and a unitary $Q$ from $\gamma$ and $A^c_0$ to $A^c_1$, respectively.

Let $\partial_{A_0}$ be the set of edges from $A_0$ to the bulk (since each boundary node has degree exactly $1$, it corresponds to exactly one edge connecting to the bulk). For $e \in \partial_{A_0}$, we cut it in the middle to get $(\ell_{e},r_e)$. Similarly, we do this for every edge $e \in \partial_{A^c_0}$. Now, by observing the resulting tensor network, we have two unitaries $U_A$ and $U_{A^c}$ such that
\[
\State(T) = (U_A \otimes U_{A^c}) \left( \bigotimes_{i \in \left([t] \cup \partial_{A_0} \cup \partial_{A^c_0}\right) } \frac{1}{\sqrt{q}} \sum_{j \in [q]} \spz{j}_{\ell_i}\spz{j}_{r_i}  \right)
\]
where $U_A$ maps $\{\ell_i\}_{i \in \left([t] \cup \partial_{A_0}\right)}$ to $A_1$, and $U_{A^c}$ maps $\{r_i\}_{i \in \left([t] \cup \partial_{A^c_0}\right)}$ to $A^c_1$.

In particular, we consider the following mapping
\[
A_T : \spz{\psi} \mapsto (U_A \otimes U_{A^c}) \left( \spz{\psi} \otimes \bigotimes_{i \in \left( \partial_{A_0} \sqcup \partial_{A^c_0} \right) } \frac{1}{\sqrt{q}} \sum_{j \in [q]} \spz{j}_{\ell_i}\spz{j}_{r_i}  \right),
\]
which replaced the first $t$ EPR pairs \[q^{-t/2} \cdot \bigotimes_{i \in [t]} \sum_{j \in [q]} \spz{j}_{\ell_i}\spz{j}_{r_i}\] by an input state $\spz{\psi}$.

\paragraph*{The pseudoentangled holographic states.} Let $\bF\colon [q] \to \{0,1\}$ be a quantum-secure pseudorandom function and $\bP \colon [q] \to [q]$ be a quantum-secure pseudorandom permutation. Now, we are ready to describe our families of pseudoentangled holographic states. 

Let $\caD_{\sf full}$ and $\caD_{\sf subset}$ be the two distributions from~\cref{cor:pseudo-ent} with $k = \sqrt{q}$. We first define two distributions $\caG_{T,\gamma}$ and $\caH_{T,\gamma}$ over holographic states, as follows:
\[
\caG_{T,\gamma} \colon A_T\left( \bigotimes_{i \in [t]} \spz{\phi_i} \right), \text{where $\spz{\phi_1},\dotsc,\spz{\phi_t} \sim \caD_{\sf full}$}
\]
and
\[
\caH_{T,\gamma} \colon A_T\left( \bigotimes_{i \in [t]} \spz{\phi_i} \right), \text{where $\spz{\phi_1},\dotsc,\spz{\phi_t} \sim \caD_{\sf subset}$}.
\]

\newcommand{\caTG}{\mathcal{TG}}
\newcommand{\caTH}{\mathcal{TH}}

In terms of tensor network, $\caG_{T,\gamma}$ corresponds to a distribution over tensor networks, denoted by $\caTG_{T,\gamma}$, that is obtained by, for each $\mu \in [t]$, replacing the edge $e_\mu$ in $T$ by the (random) tensor $ \frac{1}{\sqrt{q}} \sum_{i \in [q]} (-1)^{f_\mu(i)} \spz{i}\spz{i}$ (where $f_\mu \sim \bF$).

Similarly, $\caH_{T,\gamma}$ corresponds to a distribution over tensor networks, denoted by $\caTH_{T,\gamma}$, that is obtained by, for each $\mu \in [t]$, replacing the edge $e_\mu$ in $T$ by the (random) tensor $ \frac{1}{\sqrt{|S|}} \sum_{i \in S} (-1)^{f_\mu(i)} \spz{i}\spz{i}$ (where $S = \{ p_\mu(i) : i \in [\sqrt{q}] \}$, $p_\mu \sim \bP$, and $f_\mu \sim \bF$).

We are finally ready to define the two distributions over holographic states from~\cref{theo:PSE-general-apx}, $\caG_{G_n, \gamma}$ and $\caH_{G_n, \gamma}$, as mixed distributions of $\caG_{T,\gamma}$ and $\caH_{T,\gamma}$ over good $T \sim \caTg_{G_n}$, respectively.

\paragraph*{Computational indistinguishability.} Now we have to establish that the families $\caG_{G_n, \gamma}$ and $\caH_{G_n, \gamma}$ are computationally indistinguishable. 

\begin{lemma}\label{lm:pe_link-planar-ind}
	$\caG_{G_n, \gamma}$ and $\caH_{G_n, \gamma}$ are computationally indistinguishable.
\end{lemma}

\begin{proof}
It suffices to show that for any good $T \sim \caTg_{G_n}$, $\caG_{T,\gamma}$ and $\caH_{T,\gamma}$ are computationally indistinguishable against polynomial-time quantum adversaries given a polynomial number of copies.

We consider the following two distributions
\[
\caD_{\sf full}^{\otimes t} \colon \bigotimes_{i \in [t]} \spz{\phi_i}, \text{where $\spz{\phi_1},\dotsc,\spz{\phi_t} \sim \caD_{\sf full}$}
\]
and
\[
\caD_{\sf subset}^{\otimes t} \colon \bigotimes_{i \in [t]} \spz{\phi_i}, \text{where $\spz{\phi_1},\dotsc,\spz{\phi_t} \sim \caD_{\sf subset}$}.
\]

By~\cref{cor:pseudo-ent} and a standard hybrid argument, we know that these two distributions are computationally indistinguishable by polynomial-time quantum algorithms given a polynomial number of samples. 

Let $m \le \poly(n)$. By the definitions of $\caG_{T,\gamma}$ and $\caH_{T,\gamma}$, we know that
\[
\Ex_{\spz{\phi} \sim \caG_{T,\gamma}} \left[ \optk{\phi}{m} \right] = (A_T)^{\otimes m}\left( \Ex_{\spz{\phi} \sim \caD_{\sf full}^{\otimes t}} \left[ \optk{\phi}{m} \right] \right)
\]
and
\[
\Ex_{\spz{\phi} \sim \caH_{T,\gamma}} \left[ \optk{\phi}{m} \right] = (A_T)^{\otimes m}\left( \Ex_{\spz{\phi} \sim \caD_{\sf subset}^{\otimes t}} \left[ \optk{\phi}{m} \right] \right).
\]

From the discussions above, we know that
\[ 
\Ex_{\spz{\phi} \sim \caD_{\sf full}^{\otimes t}} \left[ \optk{\phi}{m} \right]
\]
and
\[
\Ex_{\spz{\phi} \sim \caD_{\sf subset}^{\otimes t}} \left[ \optk{\phi}{m} \right]
\]
are indistinguishable against polynomial-time quantum adversaries. Since $(A_T)^{\otimes m}$ is polynomial-time computable, it follows that $\Ex_{\spz{\phi} \sim \caG_{T,\gamma}} \left[ \optk{\phi}{m} \right]$ and $\Ex_{\spz{\phi} \sim \caH_{T,\gamma}} \left[ \optk{\phi}{m} \right]$ are also indistinguishable against polynomial-time quantum adversaries.
\end{proof}

\subsubsection{Approximate RT entanglement scaling}\label{sec:pe_link-planar-rt}

We also need to establish the approximate RT-formula for $\caG_{G_n, \gamma}$ and $\caH_{G_n, \gamma}$. The following lemma will also be useful.
\begin{lemma}\label{lm:exp-markov}
	Let $\caD$ be a distribution over quantum states. Fix $A \subseteq [n]$, $Z \in \R$, and assume that
	\[
	\Ex_{\rho \sim \caD} \left[e^{-S_A(\rho)}\right] \le Z\mathdot
	\]
	Then for all $\tau \in (0,1)$, with probability $1-\tau$ over $\rho \sim \caD$, we have
	\[
	S_A(\rho) \ge - \ln Z - \ln\tau^{-1}\mathdot
	\]
\end{lemma}
\begin{proof}
	By Markov inequality, we have
	\[
	\Pr_{\rho \sim \caD} \left[e^{-S_A(\rho)} \ge Z / \tau\right] \le \tau \mathdot
	\]
	This translates to
	\[
	\Pr_{\rho \sim \caD} \left[S_A(\rho) \le - \ln Z - \ln\tau^{-1} \right] \le \tau \mathdot \tag*{\qedhere}
	\]
\end{proof}

\paragraph*{$\caG_{G_n, \gamma}$ has holographic entropy structure approximated by $G_n$.} We first establish that $\caG_{G_n, \gamma}$ satisfies the RT-formula approximately with high probability with respect to graph $G_n$.

To show this, we need the following lemma, which can be derived using the same method from~\cite{Hayden_2016}.

\begin{lemma}\label{lm:part-G}
    Let $A \subseteq [n]$, $A^c = [n] \setminus A$ and $V$ be the vertex set of $G_n$.
    \[
    \Ex_{\spz{\phi} \sim \caG_{G_n, \gamma},~\rho = \opt{\phi}} \left[e^{-S_2(\rho_A)} \right] \le \left(1 + n^{-\omega(1)} \right) \cdot \sum_{A \subseteq S \subseteq V \setminus A^c} e^{-\Weight_{H_n}(S)}.
    \]
\end{lemma}

Applying~\cref{lm:bound-partition} with $\lambda = \ln q$ to graph $\WT{G}$ (note that $\dmax \le O(n + n_{bulk}) \le \poly(n)$), we have
\[
\Ex_{\spz{\phi} \sim \caG_{G_n, \gamma},~\rho = \opt{\phi}} \left[e^{-S_2(\rho_A)} \right] \le n^{O(\mc)} \cdot e^{-\mc \cdot \ln q} \mathdot
\]
It then follows from~\cref{lm:exp-markov} and the fact that $\lambda = \omega(\log n)$ that $\caG_{G_n, \gamma}$ has holographic entropy structure approximated by $G_n$.

\paragraph*{$\caH_{G_n, \gamma}$ satisfies the RT-formula approximately with high probability.} Next we move to $\caH_{G_n, \gamma}$. Let $H_n$ be the weighted graph obtained by changing the weights of the edges $e_1,\dotsc,e_t$ in $G_n$ from $\ln q$ to $\frac{1}{2} \cdot \ln q$.

Using the same method, we can also show the following lemma.

\begin{lemma}\label{lm:part-H}
    Let $A \subseteq [n]$, $A^c = [n] \setminus A$ and $V$ be the vertex set of $H_n$.
    \[
    \Ex_{\spz{\phi} \sim \caH_{G_n, \gamma},~\rho = \opt{\phi}} \left[e^{-S_2(\rho_A)} \right] \le \left(1 + n^{-\omega(1)} \right) \cdot \sum_{A \subseteq S \subseteq V \setminus A^c} e^{-\Weight_{H_n}(S)}.
    \]
\end{lemma}

Let $\#_A(t) = \#_{A,A^c}(\WT{G}_n,t)$ be the number of cuts between $A$ and $A^c$ in $\WT{G}_n$, and $\textsf{mc} = \mincut_{A,A^c}(\WT{G}_n)$. (Note that $\WT{H}_n $ is identical to $\WT{G}_n$.) 

First, all cuts in $H_n$ have weight $t/2 \cdot \ln q$ for some $t \in \posN$. Moreover, since the weight of an edge is either unchanged or reduced to $\ln q /2$ from $\ln q$, we note that a cut with weight $t/2 \cdot \ln q$ in $H_n$ has size between $\ceil{t/2}$ and $t$ in $\WT{G}_n$. Therefore, we have:
\begin{equation}\label{eq:bound-cut-B}
    \#_{A,A^c}(H_n, t/2 \cdot \ln q) \le \sum_{z = \ceil{t/2}}^{t} \#_A(z) \le \sum_{z = \ceil{t/2}}^{t} n^{O(z)} \le n^{c_0 \cdot t},
\end{equation}
where $c_0$ is a large constant, the second inequality above follows from~\cref{lemma:bounding-min-cuts} and~\cref{lemma:bounding-cuts} (note that both $\dmax$ and $n_{\sf f}$ are bounded by $\poly(n)$).

\newcommand{\mch}{\textsf{mch}}

Let $\mch = \mincut_{A,A^c}(H_n)$ and $\mu = \mch / \ln q$. It follows that
\begin{align*}
    \sum_{ A \subseteq S \subseteq V \setminus A^c} e^{-\Weight_{H_n}(S)} 
    \le&\sum_{t=0}^{+\infty} \#_{A,A^c}(H_n, t/2 \cdot \ln q + \mch) \cdot e^{ \mch - t/2 \cdot \ln q} \\
    \le&e^{-\mch} \cdot \left(\sum_{t=0}^{+\infty} q^{-t/2} \cdot \#_{A,A^c}(H_n, (t/2 + \mu) \cdot \ln q)\right)\\
    \le&e^{-\mch} \cdot \left(\sum_{t=0}^{+\infty} q^{-t/2} \cdot n^{c_0(t + 2\mu)} \right) \tag{by~\eqref{eq:bound-cut-B}}\\
    \le&e^{-\mch} \cdot n^{2 \cdot c_0 \cdot \mu} \cdot \left(\sum_{t=0}^{+\infty} q^{-t/2} \cdot n^{c_0 \cdot t} \right)\\
    \le&2 \cdot e^{-\mch} \cdot n^{2 \cdot c_0 \cdot \mu}\mathdot
\end{align*}

Similarly to the case of $\caG_{G_n, \gamma}$, applying~\cref{lm:exp-markov} and noting that $\mu \ln n = o(\mch)$ finishes the proof. This completes the proof of~\cref{theo:PSE-general-apx}.

\subsubsection{Proof of~\cref{lm:part-G} and~\cref{lm:part-H}}

In the following, we only prove~\cref{lm:part-H} since~\cref{lm:part-G} can be proved in exactly the same way. The below is essentially identical to the argument from~\cite{Hayden_2016}. See also~\cite[Appendix~B]{nezami2020multipartite} for a succinct presentation of the argument from~\cite{Hayden_2016} when applied to random stabilizer tensor networks. In the following, we will follow the proof from \cite[Appendix~B]{nezami2020multipartite}.

\begin{proofof}{\cref{lm:part-H}}
    
    Let $\caT_{H_n,\gamma}$ be the distribution of tensor networks obtained by replacing (1) each bulk node $u$ by an independent uniformly random stabilizer state with $d_u$ legs and $q$ bond dimensions scaled by a factor of $q^{d_u / 4}$ and (2) each edge $e \in \gamma$ by an independent tensor drawn from $\caD_{\sf subset}$ from~\cref{cor:pseudo-ent} with $k = \sqrt{q}$.
    
    \newcommand{\Contract}{\textsf{Contract}}
        
    Let $V_{b}$ denote the set of all bulk vertices from $G$, and $E_b$ denote the set of all bulk edges (that is, edges connecting bulk vertices). We also let $E_{\partial}$ be the set of edges connecting bulk nodes to boundary nodes, and $V_{\partial}$ be the set of boundary nodes. Let $\spz{V_u}$ be the random stabilizer tensors at node $u$. 
    
    We define the following unnormalized state
    \[
    \spz{\Psi} = \left( \bigotimes_{u \in V} \rpz{V_u} \right) \left( \bigotimes_{e \in E_b \setminus \gamma} \spz{e} \otimes \bigotimes_{e \in \gamma} \spz{e}\right) \mathcomma
    \]
    where $\spz{e} \sim \caD_{\sf subset}$ for every $e \in \gamma$, and $\spz{e} = \frac{1}{\sqrt{q}} \sum_{i \in [q]} \spz{i}\spz{i}$ for every $e \in E_b \setminus \gamma$. We also write $\Psi = \opt{\Psi}$ and $\rho = \Psi / \TR(\Psi)$.
    
    We note that $\caH_{G_n,\gamma}$ can be obtained by (1) drawing $\spz{V_u}$ for each $u \in V_{b}$, conditioning on the event that all $\spz{V_u}$ are perfect. (2) drawing $\spz{e} \sim \caD_{\sf subset}$ for every $e \in \gamma$, output $\rho$.
    
    First, let $N_{b} = \sum_{u \in V_{b}} d_u$ and $N_\partial = |V_\partial|$. Let $D_u = q^{d_u}$ . Since random stabilizer states form a projective $2$-design, we have 
    \[
    \Ex\left[\opt{V_u}\right] = I / D_u\quad\text{and}\quad \Ex\left[\optk{V_u}{2}\right] = \frac{I + F_u}{D_u \cdot (D_u + 1)}\mathcomma
    \]    
    where $I$ denotes the identity operator and $F_u$ denotes the swap operator on two copies of the Hilbert space of vertex $u$.

	From which we have
    \[
    \Ex[\TR(\Psi)] = \Ex\left[ \left( \bigotimes_{u \in V_b} \opt{V_u} \right) \left( \bigotimes_{e \in E_b} \opt{e}\right) \right] = q^{-N_b + N_\partial}\mathdot
    \]
    
    Indeed, we can also show that conditioning on all $\spz{V_u}$ being perfect, we have $\TR(\Psi) = q^{-N_b + N_\partial}$ exactly.
    
    For $S \subseteq V_b$, let $\partial S$ denote the set of edges from $E_b$ with exactly one endpoint in $S$. We also have
    \begin{align*}
        \Ex\left[ \TR(\Psi^2_A) \right] &= \TR\left(\Ex\left[\Psi^{\otimes 2}\right] F_A \right) \tag{$F_A$ is swap operator on the $A$ part of the two copies of $\Psi$}\\
        &= \frac{1}{\prod_{u \in V_b} D_u (D_u + 1)} \TR\left[ \left(\prod_{e \in E} \optk{e}{2} \right) \left( \prod_{u \in V_b} (I + F_u) \right) F_A  \right]\\
        &\le q^{-2 N_b} \sum_{S \subseteq V_b} \prod_{e=(u,v) \in \partial S} \TR\left[ \optk{e}{2} F_u \right] \prod_{(u,v) \in  E_{\partial}} q^{2 - \mathbb{1}_{\{(u \in S) \ne (v \in A)\}}}                                        \tag{for $(u,v) \in E_{\partial}$, we assume $u \in V_{b}$ and $v \in V_{\partial}$}\\
        &\le q^{-2 N_b + 2 N_\partial} \sum_{S \subseteq V_b} \prod_{e=(u,v) \in \partial S} \TR\left[ \optk{e}{2} F_u \right] \prod_{(u,v) \in  E_{\partial}} q^{- \mathbb{1}_{\{(u \in S) \ne (v \in A)\}}}\\
        &\le q^{-2 N_b + 2 N_\partial} \sum_{S \subseteq V_b} e^{-\Weight_{H_n}(S \cup A)}\mathdot
    \end{align*}
    
    Let $\caE$ be the event that all $\spz{V_u}$ are perfect. By~\cref{thm Random stabilizer tensors are perfect}, we have $\Pr[\caE] \ge 1 - n^{-\omega(1)}$.
    
    Therefore, we have
    \begin{align*}
        \Ex_{\spz{\phi} \sim \caH_{G_n, \gamma},~\rho = \opt{\phi}} \left[e^{-S_2(\rho_A)} \right] &= \frac{1}{q^{-2 N_b + 2 N_\partial}} \cdot \Ex\left[ \TR(\Psi^2_A) | \caE \right]\\
        &\le \frac{1}{\Pr[\caE]} \cdot    \frac{1}{q^{-2 N_b + 2 N_\partial}} \cdot \Ex\left[ \TR(\Psi^2_A) \right]\\
        &\le \left(1 + n^{-\omega(1)} \right) \cdot \sum_{S \subseteq V_b} e^{-\Weight_{H_n}(S \cup A)}\mathdot \tag*{\qedhere}
    \end{align*}
\end{proofof}

\subsubsection{Extension to public-key pseudoentangled holographic states}\label{sec:pe_link-planar-public-key}

Finally, we state the extension of~\cref{theo:PSE-general-apx} to the public-key version. We omit the proof here since it is identical to that of the tree tensor network case.

\begin{theorem}[Public-key pseudoentangled holographic states over planar graph]\label{theo:PSE-pk-planar}
	Let $\eps \in (0,1)$. Let $G = \{G_n\}_{n \in \N}$ be a family of nice bulk geometry graphs, where $G_n$ has $n$ boundary nodes and has edge weight $\ln q_n \ge \omega(\ln n)$ for some prime power $q_n$. Let $A \subseteq [n]$ be a continuous segment on the boundary of $G_n$ and let $\gamma = \{e_1,\dotsc,e_t\}$ be a minimum cut between $A$ and $A^c = [n] \setminus A$ on $G$. Assuming the standard LWE assumption. The following holds:
			
	\begin{itemize}
		\item There are two ensembles of quantum states $\caD_{\sf low}$ and $\caD_{\sf high}$ that constitute a public-key pseudoentangled holographic state ensemble with exact entropy structure\footnote{on graphs with the same set of edges as $G$ but potentially different weights} and gap $n^{\eps}$ vs. $\Omega(n)$ w.r.t. to cut $S$.
	\end{itemize}
\end{theorem}

\section{Relation between our work and the (strong) Python's lunch conjecture}\label{sec:pythons lunch}

\subsection{The (strong) Python's lunch conjecture}

For the \emph{operator reconstruction} version of implementing the AdS/CFT dictionary there exist a number of efficient algorithms that function in certain, fixed geometries. 
For example, the HKLL procedure~\cite{Hamilton_2006} can efficiently implement operator reconstruction for bulk operators lying in the causal wedge of some boundary region.
A recent follow-up~\cite{Engelhardt_2021} extends the domain of validity of HKLL to bulk operators that lie outside the outermost extremal surface associated to a boundary region.
Both these procedures assume the geometry of the bulk is fixed and known.
More precisely, these procedures work by calculating a `smearing function' which depends on first solving the bulk equations of motion (and therefore presumes a fixed bulk geometry).
The boundary operator $\phi_{\textrm{CFT}}$ dual to some bulk operator $\phi_{\textrm{AdS}}$ is then given by integrating the product of the smearing function and certain primary operators in the CFT (which are found using the extrapolate dictionary~\cite{banks1998adsdynamicsconformalfield}) over the boundary region that is space-like separated from the bulk point at which $\phi_{AdS}$ acts.
The resulting $\phi_{\textrm{CFT}}$ corresponds simply to time evolution under the local CFT Hamiltonian, and can therefore be implemented efficiently.

Studying when operator reconstruction can be carried out efficiently led to the Python's lunch conjecture~\cite{brown2019pythonslunchgeometricobstructions}.
The conjecture posits that operator reconstruction is exponentially complex if there exist locally (but not globally) minimal surfaces in the bulk (giving rise to a `Python's lunch geometry' - see \cref{fig:pythons lunch}).
The initial evidence for the conjecture arises from tensor network toy models of the duality~\cite{Pastawski_2015,Hayden_2016}.
In these toy models a Python's lunch geometry corresponds to a map from bulk to boundary which involves post-selection, and it is argued that such mappings generically lead to complex boundary operators. 
The \emph{strong} Python's lunch conjecture further posits that such geometries are the \emph{only} source of exponential complexity in operator reconstruction~\cite{Engelhardt_2021}.

In~\cite{Engelhardt_2022} the results of~\cite{bouland2019computationalpseudorandomnesswormholegrowth} were analysed with respect to the Python's lunch conjecture.
It was argued that the geometries studied in~\cite{bouland2019computationalpseudorandomnesswormholegrowth} contain Python's lunches.
These Python's lunches are not immediately apparent in the geometry of~\cite{bouland2019computationalpseudorandomnesswormholegrowth}, but appear once the randomness in the construction is treated as mixedness in the bulk degree of freedom, as is argued must be done in~\cite{Engelhardt_2022}.

\subsection{Do our constructions contain a Python's lunch?}

In order to analyse our constructions in terms of the Python's lunch conjecture we will define precisely what a Python's lunch means in the tensor network setting.

\begin{definition}
For each tensor in a HQECC let the parent legs of the tensor be legs which are contracted with a tensor which is one level closer to the centre of the tessellation. Let the children legs of the tensor be legs which are contracted with a tensor which is one level closer to the boundary of the tessellation.
\end{definition}

Note that every leg in every tensor in a HQECC constructed from Coxeter polytopes is either a parent leg, a child leg, or an uncontracted leg~\cite{Kohler_2019}. 
In this section we will restrict our attention to HQECC with this property to make the analysis concrete.
The uncontracted legs can be split into bulk legs and boundary legs:

\begin{definition}
Uncontracted legs in HQECC are boundary degrees of freedom if they are children legs of the final layer of tensors in the network. Otherwise they correspond to bulk degrees of freedom.
\end{definition}

\begin{definition}
A HQECC contains a Python's lunch iff there exists a tensor in the network that has more `input legs' (parent legs plus bulk legs) than `output legs' (children legs and boundary legs).
\end{definition}

If a tensor in a HQECC has more `input legs' than `output legs'
this gives rise to the `bulge' geometry that defines a Python's lunch (see \cref{fig:pythons lunch}).

\begin{figure}
\centering
\includegraphics[scale=0.7]{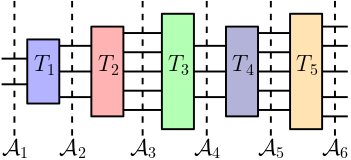}
\caption{Viewing this tensor network as a map from left to right we see that $T_3$ has more inputs than outputs, and this leads to $\mathcal{A}_4$ being a locally minimal cut in the tensor network, while the true minimum cut is $\mathcal{A}_1$.}\label{fig:pythons lunch}
\end{figure}

Care has to be taken when deciding if a tensor network contains a Python's lunch or not. 
It is not enough to simply pick a state in the ensemble of possible boundary states and demonstrate that that particular state does not contain a Python's lunch.
Instead we must consider the maximally mixed state over the entire ensemble~\cite{Engelhardt_2022}.
This means that the randomness in our constructions has to be taken into account.
In particular, in cases where we have a key this key should be treated as a bulk degree of freedom, and the uncertainty in the key as mixedness in that bulk degree of freedom.

For the construction based on pseudoentangled link states this necessarily leads to a Python's lunch geometry when considering the full ensemble of possible boundary states.

The situation for the construction based on low-entangling PRUs is more subtle.
At first it appears that treating the key to the PRU as an input necessarily leads to a Python's lunch geometry, since each small PRU is now mapping from $\omega(\log n ) + \log(k)$ qubits to $\omega(\log n )$ qubits.
However, we note that the proof of the brickwork PRU construction can be extended to give rise to a brickwork pseudorandom isometry (PRI) construction (see \cref{sec:PRI}).
This means we can replace the PRUs in our construction with PRIs without changing any of the conclusions from \cref{sec:pe_pru}.
In particular, we can choose the parameters of the PRIs such that every leg in the tensor network contains at least as many output legs as input legs.
It appears that this has removed the Python's lunch from the construction.
However, it should be noted that the individual tensors in the HQECC are no longer isometric.
This is because each PRI is an isometry, but the map that takes the key as input and implements a particular isometry from the ensemble is not itself isometric.
Decomposing the tensors into smaller components such that every component is an isometric tensor will require the introduction of ancilla registers, and a Python's lunch geometry will appear on this smaller scale.

In both cases the Python's lunch geometry can be avoided if we settle for practical security, as opposed to provable security.
To see what this means, note that while cryptographic proofs of security require the notion of indistinguishable ensembles, in reality any practical implementation of cryptography refers to a single instance.
While we cannot prove that decrypting a single instance is hard, in practise we find that it is.
Working within this paradigm, we could remove the need for randomness in our constructions, and argue that while we cannot prove that our constructions remain hard for a fixed value of the key, they are likely to.
This removes the need to take into account the randomness of the key as a bulk degree of freedom, and we can argue for hardness of geometry reconstruction in the absence of a Python's lunch.
We note that in this setting of `practical security' for state reconstruction, operator reconstruction is easy (as predicted by the strong Python's lunch conjecture). 
Therefore this provides an example of a situation where there is evidence for a gap in complexity between operator reconstruction and state reconstruction.

We note that if there was a public key version of the PRI construction this would imply a construction which is both provably secure, and does not have a Python's lunch.
This is because in this case we could treat the key as part of the input without needing to `throw away' information about the key to obtain the boundary state.
It is this act of `throwing away the key' that leads to a Python's lunch in the brickwork PRI construction with randomness.
The model from a pseudoentangled link state does work with a public key, but here have a Python's lunch anyway because the randomness is associated to an edge in the tensor network as opposed to a tensor, and there is no analogue of replacing a unitary with an isometry to increase the output space of a link state.
However, the existence of public key holographic pseudoentanglement suggests that there is no fundamental barrier to constructing such a scheme, which we leave open as an interesting avenue for future research.

Evidently the question of whether or not our constructions contain a Python's lunch depends subtly on the exact setting of the question, and the notion of security required.

\bibliographystyle{alpha}
\bibliography{main}

\appendix

\newpage

\section{Min-cuts on Bulk Geometry Graphs}

In this section, we discuss some properties of bulk geometry graphs used in the main body of the paper.

\subsection{Bounding the Number of Cuts} \label{sec:bound-the-number-of-cuts}

In the following, for simplicity we assume that $G$ has unit weight. Let $A \subseteq [n]$ and $A^c = [n] \setminus A$ be its complement in the boundary. We assume both $A$ and $A^c$ are non-empty. We also let $\#(t) = \#_{A,A^c}(G,t)$ be the number of cuts between $A$ and $A^c$ on $G$.

\subsubsection{A Cut in $G$ as a collection of paths and cycles in $G^*$} 
Let $\nseg$ be the number of continuous segments of the boundary (note that $\nseg$ is always even). Let $\{ s_1,\dotsc,s_{\nseg} \}$ be the set of faces that is connected the endpoints of these segments. Our first observation is that a cut $S$ between $A$ and $A^c$ (i.e., $A \subseteq S \subseteq V \setminus A^c$) induces a collection of edge-disjoint paths that connect these $ s_1,\dotsc,s_{\nseg} $ in pairs, as well as some other cycles; see~\autoref{fig:paths-and-cycles-on-graph}.

We also note that a minimum cut between $A$ and $A^c$ on $G$ does not contain cycles, since cycles can always to decrease the size of the cut.

\begin{figure}[ht]
	\resizebox{0.45\linewidth}{!}{		
		\begin{tikzpicture}[scale=1.0, line1/.style={},
			line2/.style={},
			line3/.style={},
			]
			% Nodes
			\foreach \i in {1,12,11,5,6,7}{
				\pgfmathtruncatemacro{\tmp}{\i*30}
				\node[circle,draw=red,fill=white,line width=1.5pt] (leaf\i) at (\tmp:3.3cm) {\i};
			}
			
			\foreach \i in {2,3,4,8,9,10}{
				\pgfmathtruncatemacro{\tmp}{\i*30}
				\node[circle,draw=blue,fill=white,line width=1.5pt] (leaf\i) at (\tmp:3.3cm) {\i};
			}
			
			\foreach \x in {1,...,3}{
				\foreach \y in {1,...,3}{
					\pgfmathtruncatemacro{\tmp}{(\x-1)*3+\y}
					\ifnum\tmp=1
					\node[rectangle,draw=red,fill=white,line width=1.5pt] (center-\x-\y) at (1.5*\y-3, 3-1.5*\x)  {$C_{\tmp}$};
					\else
					\ifnum\tmp=5
					\node[rectangle,draw=red,fill=white,line width=1.5pt] (center-\x-\y) at (1.5*\y-3, 3-1.5*\x)  {$C_{\tmp}$};
					\else
					\ifnum\tmp=6
					\node[rectangle,draw=red,fill=white,line width=1.5pt] (center-\x-\y) at (1.5*\y-3, 3-1.5*\x)  {$C_{\tmp}$};
					\else
					\node[rectangle,draw=blue,fill=white,line width=1.5pt] (center-\x-\y) at (1.5*\y-3, 3-1.5*\x)  {$C_{\tmp}$};
					\fi\fi\fi
					
				}
			}
			\newcommand{\drawth}[4]{
				% Calculate the centroid
				\path (barycentric cs:#1=1,#2=1,#3=1) coordinate (centroid);
				% Draw the centroid
				\draw[line width=2] (#1) -- (#2);
				\draw[line width=2] (#1) -- (#3);
				\node[inner sep=0pt, outer sep=0pt] (face-#4) at (centroid) {};
			}			
			
			\newcommand{\drawf}[5]{
				% Calculate the centroid
				\path (barycentric cs:#1=1,#2=1,#3=1,#4=1) coordinate (centroid);
				% Draw the centroid
				\draw[line width=2] (#1) -- (#3);
				\draw[line width=2] (#2) -- (#4);
				\node[inner sep=0pt, outer sep=0pt] (face-#5) at (centroid) {};
			}
			% Macro to draw the line from the reflection of C across AB to C
			
			\newcommand{\drawReflectionLine}[3]{
				% Calculate the reflection of point #3 across the line defined by #1 and #2
				% This uses the 'calc' library for coordinate calculations
				\coordinate (reflected) at ($2*(#1)!0.5!(#2)-(#3)$);
				
				% Draw the reflection point
				
				% Draw line from the reflected point to C
				\draw [dashed, violet, line width=2] (reflected) -- (#3);
			}
			
			\foreach \x in {1,...,3}{
				\foreach \y in {1,...,3}{
					\pgfmathtruncatemacro{\a}{\x+1}
					\pgfmathtruncatemacro{\b}{\y+1}
					\pgfmathtruncatemacro{\tmp}{13+(\x-1)*2+(\y-1)}
					\ifnum \x < 3
					\draw[line width=2] (center-\x-\y) -- (center-\a-\y);
					\ifnum \y < 3
					\drawf{center-\x-\y}{center-\a-\y}{center-\x-\b}{center-\a-\b}{\tmp}
					\fi
					\fi
					\ifnum \y < 3
					\draw[line width=2] (center-\x-\y) -- (center-\x-\b);
					\fi
				}
			}
			
			\drawth{center-1-1}{leaf4}{leaf5}{1};
			\drawf{center-1-1}{center-1-2}{leaf4}{leaf3}{2};
			\drawf{center-1-2}{center-1-3}{leaf3}{leaf2}{3};
			\drawth{center-1-3}{leaf2}{leaf1}{4};
			\drawf{center-1-3}{center-2-3}{leaf1}{leaf12}{5};
			\drawf{center-2-3}{center-3-3}{leaf12}{leaf11}{6};
			\drawth{center-3-3}{leaf11}{leaf10}{7};
			\drawf{center-3-3}{center-3-2}{leaf10}{leaf9}{8};
			\drawf{center-3-2}{center-3-1}{leaf9}{leaf8}{9};
			\drawth{center-3-1}{leaf7}{leaf8}{10};
			\drawf{center-3-1}{center-2-1}{leaf7}{leaf6}{11};
			\drawf{center-2-1}{center-1-1}{leaf6}{leaf5}{12};
			
			\drawReflectionLine{leaf4}{leaf5}{face-1};
			\draw [dashed, violet, line width=2] (face-1) -- (face-2);
			\draw [dashed, violet, line width=2] (face-2) -- (face-13);
			\draw [dashed, violet, line width=2] (face-13) -- (face-14);
			\draw [dashed, violet, line width=2] (face-14) -- (face-5);
			\draw [dashed, violet, line width=2] (face-5) -- (face-4);
			\drawReflectionLine{leaf2}{leaf1}{face-4};
			
			\drawReflectionLine{leaf7}{leaf8}{face-10};
			\draw [dashed, violet, line width=2] (face-10) -- (face-11) -- (face-12) -- (face-13) -- (face-15) -- (face-16) -- (face-6) -- (face-7);
			\drawReflectionLine{leaf10}{leaf11}{face-7};
			
			\begin{scope}[on background layer]
				\draw[dashed, line width=2] (0,0) circle [radius=3.3cm];
			\end{scope}
			
			% Dashed circle through leaves
		\end{tikzpicture}
	}
	\qquad
	\resizebox{0.45\linewidth}{!}{		
		\begin{tikzpicture}[scale=1.0, line1/.style={},
			line2/.style={},
			line3/.style={},
			]
			% Nodes
			\foreach \i in {5,6,7,1,12,11}{
				\pgfmathtruncatemacro{\tmp}{\i*30}
				\node[circle,draw=red,fill=white,line width=1.5pt] (leaf\i) at (\tmp:3.3cm) {\i};
			}
			
			\foreach \i in {2,3,4,8,9,10}{
				\pgfmathtruncatemacro{\tmp}{\i*30}
				\node[circle,draw=blue,fill=white,line width=1.5pt] (leaf\i) at (\tmp:3.3cm) {\i};
			}
			
			\foreach \x in {1,...,3}{
				\foreach \y in {1,...,3}{
					\pgfmathtruncatemacro{\tmp}{(\x-1)*3+\y}
					\ifnum\tmp=10
					\node[rectangle,draw=red,fill=white,line width=1.5pt] (center-\x-\y) at (1.5*\y-3, 3-1.5*\x)  {$C_{\tmp}$};
					\else
					\node[rectangle,draw=blue,fill=white,line width=1.5pt] (center-\x-\y) at (1.5*\y-3, 3-1.5*\x)  {$C_{\tmp}$};
					\fi
					
				}
			}
			\newcommand{\drawth}[4]{
				% Calculate the centroid
				\path (barycentric cs:#1=1,#2=1,#3=1) coordinate (centroid);
				% Draw the centroid
				\draw[line width=2] (#1) -- (#2);
				\draw[line width=2] (#1) -- (#3);
				\node[inner sep=0pt, outer sep=0pt] (face-#4) at (centroid) {};
			}			
			
			\newcommand{\drawf}[5]{
				% Calculate the centroid
				\path (barycentric cs:#1=1,#2=1,#3=1,#4=1) coordinate (centroid);
				% Draw the centroid
				\draw[line width=2] (#1) -- (#3);
				\draw[line width=2] (#2) -- (#4);
				\node[inner sep=0pt, outer sep=0pt] (face-#5) at (centroid) {};
			}
			% Macro to draw the line from the reflection of C across AB to C
			
			\newcommand{\drawReflectionLine}[3]{
				% Calculate the reflection of point #3 across the line defined by #1 and #2
				% This uses the 'calc' library for coordinate calculations
				\coordinate (reflected) at ($2*(#1)!0.5!(#2)-(#3)$);
				
				% Draw the reflection point
				
				% Draw line from the reflected point to C
				\draw [dashed, violet, line width=2] (reflected) -- (#3);
			}
			
			\foreach \x in {1,...,3}{
				\foreach \y in {1,...,3}{
					\pgfmathtruncatemacro{\a}{\x+1}
					\pgfmathtruncatemacro{\b}{\y+1}
					\pgfmathtruncatemacro{\tmp}{13+(\x-1)*2+(\y-1)}
					\ifnum \x < 3
					\draw[line width=2] (center-\x-\y) -- (center-\a-\y);
					\ifnum \y < 3
					\drawf{center-\x-\y}{center-\a-\y}{center-\x-\b}{center-\a-\b}{\tmp}
					\fi
					\fi
					\ifnum \y < 3
					\draw[line width=2] (center-\x-\y) -- (center-\x-\b);
					\fi
				}
			}
			
			\drawth{center-1-1}{leaf4}{leaf5}{1};
			\drawf{center-1-1}{center-1-2}{leaf4}{leaf3}{2};
			\drawf{center-1-2}{center-1-3}{leaf3}{leaf2}{3};
			\drawth{center-1-3}{leaf2}{leaf1}{4};
			\drawf{center-1-3}{center-2-3}{leaf1}{leaf12}{5};
			\drawf{center-2-3}{center-3-3}{leaf12}{leaf11}{6};
			\drawth{center-3-3}{leaf11}{leaf10}{7};
			\drawf{center-3-3}{center-3-2}{leaf10}{leaf9}{8};
			\drawf{center-3-2}{center-3-1}{leaf9}{leaf8}{9};
			\drawth{center-3-1}{leaf7}{leaf8}{10};
			\drawf{center-3-1}{center-2-1}{leaf7}{leaf6}{11};
			\drawf{center-2-1}{center-1-1}{leaf6}{leaf5}{12};
			
			\drawReflectionLine{leaf4}{leaf5}{face-1};
			\draw [dashed, violet, line width=2] (face-1) -- (face-12) -- (face-11) -- (face-10);
			\drawReflectionLine{leaf2}{leaf1}{face-4};
			
			\drawReflectionLine{leaf7}{leaf8}{face-10};
			\draw [dashed, violet, line width=2] (face-4) -- (face-5) -- (face-6) -- (face-7);
			\drawReflectionLine{leaf10}{leaf11}{face-7};
			
			%\draw [dashed, violet, line width=2] (face-13) -- (face-14) -- (face-16) -- (face-15) -- (face-13);
			
			\begin{scope}[on background layer]
				\draw[dashed, line width=2] (0,0) circle [radius=3.3cm];
			\end{scope}
			
			% Dashed circle through leaves
		\end{tikzpicture}
	}
	\caption{Two cuts on $G$ between $A$ and $A^c$ (red vertices denote the set $S$ such that $A \subseteq S \subseteq V \setminus A^c$)}
	\label{fig:paths-and-cycles-on-graph}
\end{figure}
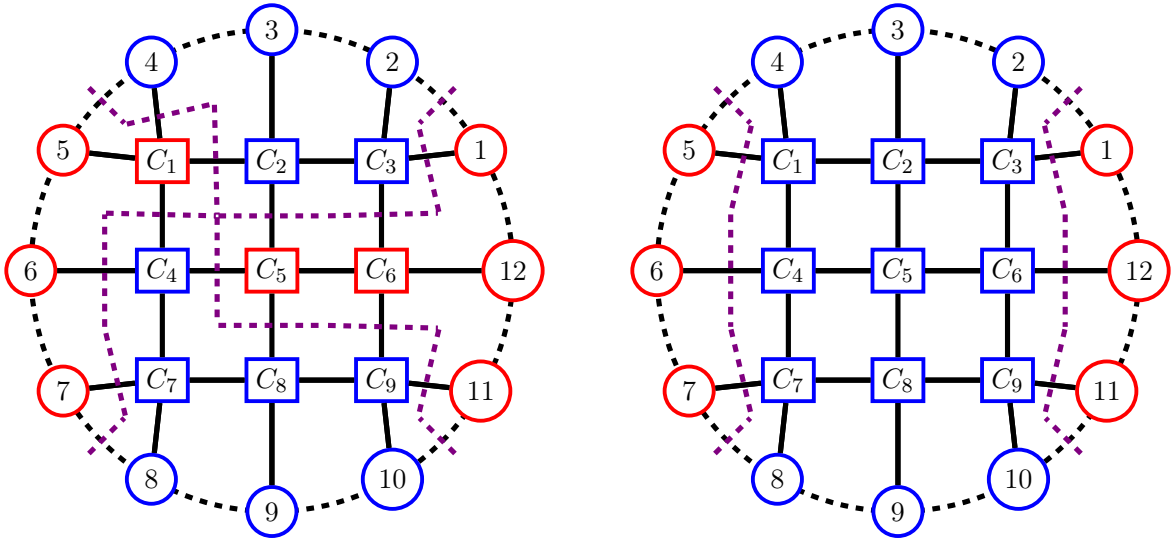

\subsubsection{Bounding the number of min-cuts in $G^*$}

Let the maximum degree of $G^*$ be $\dmax$. Let $\mc = \min_{A \subseteq S \subseteq (V \setminus A^c)}(\WT{G})$. The following fact would be helpful.

\begin{fact}\label{fact:simple-bound}
    Let $m,n \in \posN$ be such that $m \ge n$. It holds that
    \[
    \binom{m+(n-1)}{n-1} \le e^{O(m)}\mathdot
    \]
\end{fact}
\begin{proof}
    We have
    \begin{align*} 
        \binom{m+(n-1)}{n-1} &\le \binom{m+n}{n}\\
        &\le \left( e \cdot \frac{m+n}{n}\right)^{n} \tag{$ \binom{a}{b} \le \left(\frac{e\cdot a}{b}\right)^b$}\\
        &\le e^n \cdot \left( 1 + \frac{m}{n}\right)^{n}\\
        &\le e^n \cdot e^m \le e^{O(m)} \mathdot
    \end{align*}
    The last inequality holds since
    \[
    \left( 1 + \frac{m}{n}\right)^{n} \le \lim_{a \to \infty}\left( 1 + \frac{m}{a}\right)^{a} = e^{m}.\tag*{\qedhere}
    \]
\end{proof}

We have the following lemma bounding the number of min-cuts in $G^*$.

\begin{lemma}\label{lemma:bounding-min-cuts}
    Let $G$ be a bulk geometry graph such that $G^*$ has maximum degree $\dmax$. It holds that
    \[
    \#(\mc) \le (\dmax)^{\mc} \cdot e^{O(\mc)}\mathdot
    \]
\end{lemma}
\begin{proof}
    Note that $\mc \ge \nseg / 2$. We can bound the number of cuts between $A$ and $A^c$ with total size $\mc$ by bounding the number of collections of paths that connect $\{ s_1,\dotsc,s_{\nseg} \}$ in pairs, as follows:
    \[
    \binom{\nseg}{\nseg/2} \cdot (\dmax)^{\mc} \cdot \binom{\mc+(\nseg/2-1)}{(\nseg/2-1)}\mathdot
    \]
    The first term $\binom{\nseg}{\nseg/2}$ corresponds to choosing $\nseg/2$ starting points among $\nseg$ endpoints of the segments on the boundary. The last term $\binom{\mc+(\nseg/2-1)}{(\nseg/2-1)}$ corresponds to the total possible length configurations of these $\nseg/2$ paths (their lengths sum up to $\mc$).
    
    The lemma follows directly from~\autoref{fact:simple-bound} and $\binom{\nseg}{\nseg/2} \le e^{O(\mc)}$.
\end{proof}

\subsubsection{Bounding the number of cuts in $G^*$}

\newcommand{\ncyc}{n_{\sf cyc}}
\newcommand{\nf}{n_{\sf f}}

Now we move to bound the number of cuts with size larger than $\mc$.

\begin{lemma}\label{lemma:bounding-cuts}
    Let $G$ be a bulk geometry graph such that $G^*$ has maximum degree $\dmax$. For $t \in \posN$, it holds that
    \[
    \#(\mc+t) \le \nf^{t} \cdot (\dmax)^{\mc+t} \cdot e^{O(\mc+t)} \mathdot
    \]
\end{lemma}
\begin{proof}
    Let $S$ be such that $A \subseteq S \subseteq V \setminus A^c$. Recall that $\Weight_G(S)$ is the total weight of edges with exactly one endpoint contained in $S$. We wish to bound the number of sets $S$ with $\Weight_G(S) = \mc + t$.
    
    As we discussed before, $S$ induces a collection of edge-disjoint paths and cycles in $G^*$ such that the paths connect $\{ s_1,\dotsc,s_{\nseg} \}$ in pairs. We first observe that the total size (the sum of the lengths) of the cycles is at most $t$, since otherwise by removing all these cycles, we can obtain a min-cut between $A$ and $A^c$ with size less than $\mc$, a contradiction to the definition of $\mc$.
    
    In a planar graph, the number of faces $\nf$ is bounded by $O(|V|) \le O(n + n_{\sf bulk})$. To describe a cycle of length $d$, we can fix a starting face and then list the indices of all the outgoing edges. Hence, there are at most $\nf \cdot (\dmax)^{d}$ many cycles of length $d$ in $G^*$.
    
    Suppose the total size of cycles is $w \le t$. Since each cycle has at least $2$ edges, it means there are at most $\ceil{w/2}$ cycles. Suppose there are $k \le \floor{w/2}$ cycles. We can bound the number of collections of cycles with total size $w$ by
    \begin{align*}
        &\sum_{k=1}^{\floor{w/2}} \nf^{k} \cdot \binom{w+(k-1)}{k-1} \cdot (\dmax)^{w} \le O\left(\nf^{w} \cdot (\dmax)^{w} \cdot e^{O(w)} \right)\mathdot
    \end{align*}
    
    The total length of the paths is $\mc + t - w$, and we can bound the number of such collections of paths by
    \[
    (\dmax)^{\mc+t-w} \cdot e^{O(\mc+t-w)}
    \]
    similar to the proof of~\autoref{lemma:bounding-min-cuts}.
    
    Enumerating the possible sizes of cycles, we have
    \begin{align*}
        \#(\mc+t) &\le \sum_{w=0}^{t} O\left(\nf^{w} \cdot (\dmax)^{w} \cdot e^{O(w)} \cdot (\dmax)^{\mc+t-w} \cdot e^{O(\mc+t-w)} \right) \\
        &\le \nf^{t} \cdot (\dmax)^{\mc+t} \cdot e^{O(\mc+t)}\mathdot\tag*{\qedhere}
    \end{align*}
\end{proof}

\subsubsection{Upper bounding the partition function}

Let $c_0$ be a large enough absolute constant that can be used in place of the big-O notation from~\autoref{lemma:bounding-min-cuts} and~\autoref{lemma:bounding-cuts}.

Recall that $G$ has unit weight. Let $\lambda > 0$ be a parameter. We will be interested in the following partition function
\[
Z_{G}(\lambda) \coloneqq \sum_{A \subseteq S \subseteq V \setminus A^c} e^{-\Weight_G(S) \cdot \lambda}\mathdot
\]

We have the following upper bound on $Z_G(\lambda)$ when $\lambda$ is large enough.

\begin{lemma}\label{lm:bound-partition}
    Let $G$ be a bulk geometry graph such that $G^*$ has maximum degree $\dmax$. Assuming $\lambda \ge 2 \cdot \ln \left(\nf \cdot \dmax \cdot e^{c_0}\right)$, it holds that
    \[
    Z_{G}(\lambda) \le 2 \cdot e^{-\mc \cdot \lambda} \cdot (\dmax)^{\mc} \cdot e^{c_0 \cdot \mc}\mathdot
    \]
\end{lemma}
\begin{proof}
    We have
    \begin{align*}
        Z_{G}(\lambda) &= \sum_{A \subseteq S \subseteq V \setminus A^c} e^{-\Weight_G(S) \cdot \lambda}\\
        &= \sum_{t=0}^{\infty} \#(\mc + t) \cdot e^{-(\mc + t) \cdot \lambda}\\
        &\le e^{-\mc \cdot \lambda} \cdot (\dmax)^{\mc} \cdot e^{c_0 \cdot \mc} \cdot \sum_{t=0}^{\infty} \nf^{t} \cdot (\dmax)^{t} \cdot e^{c_0 \cdot t} \cdot e^{- t \cdot \lambda}\\
        &\le e^{-\mc \cdot \lambda} \cdot (\dmax)^{\mc} \cdot e^{c_0 \cdot \mc} \cdot \left[1 + \sum_{t=1}^{\infty} \left( \nf \cdot \dmax \cdot e^{c_0} \cdot e^{- \lambda} \right)^{t} \right]\mathdot
    \end{align*}
    From our assumption on $\lambda$, we have
    \[
    Z_{G}(\lambda) \le 2 \cdot e^{-\mc \cdot \lambda} \cdot (\dmax)^{\mc} \cdot e^{c_0 \cdot \mc}\mathdot\tag*{\qedhere}
    \]
\end{proof}
%\section{Poly-size bond dimension}

\section{Omitted proofs}

\subsection{Brickwork Pseudorandom isometry construction} \label{sec:PRI}

\begin{definition}[Haar isometry]
We call an isometry $\mathcal{I}:\mathbb{C}^N \rightarrow \mathbb{C}^{NM}$ a Haar isometry if $\mathcal{I}\ket{x} = U\ket{x}\ket{\hat 0}$ where $U:\mathbb{C}^{NM}\rightarrow \mathbb{C}^{NM}$ is a Haar unitary and $\ket{\hat 0} \in \mathbb{C}^M$ is an arbitrary and fixed pure state.
\end{definition}

\begin{definition}[Pseudorandom Isometry (PRI)]

\end{definition}

\begin{lemma}[Lemma 8~\cite{brickwork_t_design}]\label{lemma 8}
Let $A$, $B$, $C$ be three disjoint subsystems.
Consider a random unitary given by $V_{ABC} = U_{AB}U_{BC}$ where $U_{AB}$ and $U_{BC}$ are drawn from $\epsilon_{AB}$ and $\epsilon_{BC}$-approximate unitary $k$-designs respectively. 
Then $V_{ABC}$ is a $\epsilon$-approximate unitary $k$-design for:
\begin{equation}
1 + \epsilon = (1 + \epsilon_{AB})(1+\epsilon_{BC})\left(1 + 2 \left( \frac{k^2}{D_B} +\frac{k^2}{D_{BC}}+ \frac{k^2}{D_BD_{BC}} + \frac{\frac{k^2}{2D_{BC}}}{1-\frac{k^2}{2D_{BC}}} \right)\left(1 + \frac{k^2}{D_{AB}} \right) \right)
\end{equation}
as long as $k^2 \leq D_B$ where $D_\alpha = 2^{|{\alpha}|}$ is the Hilbert space dimension of subsystem $\alpha$.
\end{lemma}

We can use \cref{lemma 8} to prove a slightly modified version of Theorem 1~\cite{brickwork_t_design} where instead of assuming that the two layers of brickwork contain unitaries of the same size we assume that the second row of unitaries is larger, with some external inputs. \cref{fig:pri}.
These external inputs will allow us to construct a PRI instead of a PRU.

\begin{figure}
\centering
\includegraphics[scale=0.5]{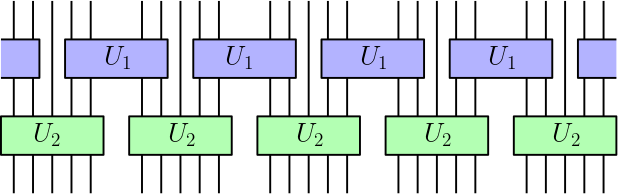}
\caption{The brickwork construction we use to construct a low-depth PRI. The first row of qubits acts on $\omega(\log n)$ qubits, while the second row of qubits acts on $\omega(\log n) + \eta'$ qubits.
One input for each unitary in the brickwork construction will be fixed to $\ket{0}$ to give a PRI.} \label{fig:pri} 
\end{figure}

\begin{lemma}[Modification of Theorem 1~\cite{brickwork_t_design}]
Consider any approximation error $\epsilon \leq 1$. Suppose each small random unitary in the first layer of the brickwork ensemble $\mathcal{E}$ is drawn from an $\frac{\epsilon}{n}$-approximate unitary $k$-design on $2 \eta$ qubits with circuit depth $d$, and each small random unitary in the second layer of the brickwork ensemble $\mathcal{E}$ is drawn from an $\frac{\epsilon}{n}$-approximate unitary $k$-design on $2 \eta + \eta'$ qubits with circuit depth $d$.
Then $\mathcal{E}$ forms an $\epsilon$-approximate unitary $k$-design on $n$ qubits with depth $2d$, whenever the local patch size satisfies $\eta \geq \log_2(nk^2/\epsilon)$ and $\eta' \geq 1$ 
\end{lemma}
\begin{proof}
We will apply \cref{lemma 8} patch-by-patch. 
Let $m$ be the number of patches of $\eta$ qubits.
Then there will be a total of $m$ small random unitaries applied. 
Let $q = 2^\eta$ and $q' = 2^{\eta'}$.
Then we have that $D_B = q$ for every application of \cref{lemma 8}.
$D_C$ will alternate between $D_C = q$ and $D_C = qq'$ (the former when we are adding a random unitary in the top layer, the latter when we are adding a random unitary in the bottom layer.
$D_A = q$ for the first application of \cref{lemma 8}, then it increases by $D_C$ for every later application. 

Therefore we have that after $m$ applications of \cref{lemma 8} the brickwork ensemble forms a $k$-design with error:

\begin{equation} \label{eq:error}
\begin{split}
\left(1+\frac{\epsilon}{n}\right)^m\left(1+ f(k,q) \right)^{\frac{m}{2}}\left(1+ g(k,q,q') \right)^{\frac{m}{2}} - 1 &\leq \exp\left(\frac{m \epsilon}{2} + \frac{m f(k,q)}{2} + \frac{mg(k,q,q')}{2}\right) - 1 \\ 
& \leq \frac{1}{\log 2} \left(\frac{m \epsilon}{2} + \frac{m f(k,q)}{2} + \frac{mg(k,q,q')}{2}  \right)
\end{split}
\end{equation}
where:
\begin{equation}
f(k,q) =  \frac{k^2}{q} +\frac{k^2}{q^2}+ \frac{k^2}{q^3} + \frac{\frac{k^2}{2q^2}}{1-\frac{k^2}{2q^2}} 
\end{equation}
and
\begin{equation}
g(k,q,q') =  \frac{k^2}{q^2} +\frac{k^2}{q^2q'}+ \frac{k^2}{q^3q'} + \frac{\frac{k^2}{2q^2q'}}{1-\frac{k^2}{2q^2q'}} 
\end{equation}

We need to show this error is less than $\epsilon$.
As in~\cite{brickwork_t_design} we take $k \geq 2$ and $n \geq 3\eta$ as otherwise the theorem holds trivially.
By assumption we have $\epsilon \leq 1$ and $q \geq nk^2/\epsilon$, giving $\eta \geq 7$ and $q \geq 128$.
Therefore the first term in \cref{eq:error} is:
\begin{equation}
\frac{m \epsilon}{n \log 2} \leq \frac{\epsilon}{7\log 2}
\end{equation}
since $ m \leq n/ \eta \leq n / 7$.

Applying $q \geq nk^2 / \epsilon$ and $q \geq 2$ to the second term in \cref{eq:error} gives:
\begin{equation}
\begin{split}
        \frac{m f(k,q)}{2 \log 2} &\leq \frac{n}{7 \log 2}\left(\frac{\epsilon}{n} + \frac{\epsilon}{nq} + \frac{\epsilon^2}{n^2q} + \frac{\frac{\epsilon}{2nq}}{1-\frac{\epsilon}{2nq}} \right)\left(1 + \frac{\epsilon}{nq} \right) \\ 
        &\leq \frac{\epsilon}{7 \log 2}\left(1 + \frac{1}{128} + \frac{1}{21 \times 128} + \frac{\frac{1}{256}}{1-\frac{1}{2 \times 21 \times 128}} \right)\left(1 + \frac{1}{21 \times 128} \right) \\
        &\leq \frac{102 \epsilon}{700 \log 2}
    \end{split}
\end{equation}

Applying $q \geq nk^2 / \epsilon$ to the third term in \cref{eq:error} gives:
\begin{equation}
\begin{split}
        \frac{m g(k,q,q')}{2 \log 2} &\leq \frac{n}{7 \log 2}\left(\frac{\epsilon}{n} + \frac{\epsilon}{2nq} + \frac{\epsilon^2}{2n^2q} + \frac{\frac{\epsilon}{4nq}}{1-\frac{\epsilon}{4nq}} \right)\left(1 + \frac{\epsilon}{nq} \right) \\ 
        &\leq \frac{\epsilon}{7 \log 2}\left(1 + \frac{1}{256} + \frac{1}{2 \times 21 \times 128} + \frac{\frac{1}{512}}{1-\frac{1}{4 \times 21 \times 128}} \right)\left(1 + \frac{1}{21 \times 128} \right) \\
        &\leq \frac{101 \epsilon}{700 \log 2}
    \end{split}
\end{equation}
Therefore the total errors is less than $\frac{303 \epsilon}{700 \log2} < \epsilon$ as required. 
\end{proof}

Finally we can prove that the overall brickwork construction is a PRI:
\begin{theorem}
Let $\mathcal{E}$ be the two-layer brickwork ensemble in \cref{fig:pri} where each small random unitary in the first layer is a $2 \eta$-qubit PRU and each small random unitary in the second layer is a $(2 \eta + \eta')$-qubit PRU, both secure against $\poly(n)$-time adversaries.
Then the ensemble of isometries given by:
\begin{equation}
V \ket{\psi} = U \ket{\psi} \ket{0}^{\otimes m}
\end{equation}
for $U \leftarrow \mathcal{E}$  where the $\ket{0}$ inputs are applied to one free input for each small PRU in the brickwork construction is a PRI secure against $\poly(n)$-time adversaries.
\end{theorem}
\begin{proof}
By Theorem 4~\cite{brickwork_t_design} unitaries from $\mathcal{E}$ are pseudorandom unitaries secure against $\poly(n)$-time adversaries.
If $V$ was distinguishable from a Haar isometry by $\poly(n)$-time adversaries this would provide a method for a $\poly(n)$-time adversary to distinguish $U$ from a Haar random unitary.
\end{proof}

\end{document}